\documentclass[review]{elsarticle}

\usepackage{hyperref}
\usepackage{amssymb}

\newdefinition{mydef}{Definition}
\newdefinition{myapp}{Approach}
\newtheorem{mytheo}{Theorem}
\newtheorem{myobs}{Observation}
\newtheorem{mycolo}{Corollary}
\newproof{myproof}{Proof sketch}

\newcommand*{\qeds}{\hfill\ensuremath{\square}}

\journal{Theoretical Computer Science}

\begin{document}

\begin{frontmatter}

\title{Towards Solving NP-Complete and Other Hard Problems Efficiently in Practice}
\date{November 2019}

\author{Mircea-Adrian Digulescu}
\address{Independent}
\ead{mircea.digulescu@gmail.com}
\address{Former: Department of Computer Science, Faculty of Mathematics and Computer Science, University of Bucharest Romania}
\ead{mircea.digulescu@my.fmi.unibuc.ro}

\begin{abstract}
Until now, Computer Scientists have concerned themselves with identifying efficient algorithms for solving the general case of some problem - that is finding one which performs well when the size of the input tends to infinity. However, this is the precise opposite of what is actually needed in practice. Effectively solving some real-world problem entails identifying an algorithm which works well for all (or some) inputs up to some fixed upper bound dictated by the concrete practical application. Such an algorithm may be distinct from the one which solves the general case. Furthermore, a general case algorithm may not exist at all or finding it might prove painstakingly hard for the human mind.  Fortunately, in practice all that is needed is one which works on the finite cases involved in the real world situations, not one which can, unaltered, solve any input correctly.

In this paper, we first introduce a theoretical framework for reasoning about finite algorithmics. It allows familiar concepts such as asymptotic complexity to be adapted to the case where the input size is bounded from above. We also present some elementary results within this theory. Secondly, we present a generic approach for automatically discovering an adequate algorithm for the finite case of some hard problem - if one exists. Thirdly, we argue why we expect the finite case of hard problems to be easier than the general case. Fourthly, we present some relevant ideas specific to three hard problems, namely 3CNF-SAT, String Compression and Integer Factorization. Fifthly, we discuss the significance of the theory and methods introduced in this paper - noting among other things that they can be used to automatically determine that either (i) $P = NP$, (ii) $P \neq NP$ or (iii) we don’t really care about the distinction for practical purposes. Finally, we present four directions for immediate further research and formulate an open question which, when answered will, for all practical purposes, decide $P=NP$. 

Enhancing the way Computer Scientists reason about hard problems is ultimately the single most important contribution we claim for this paper.
\end{abstract}


\end{frontmatter}

\section{Introduction}\label{Sec1}

Until now, we as Computer Scientists have almost exclusively concerned ourselves with finding algorithms to solve interesting problems in the general case. That is to identify a single algorithm - some fixed finite sequence of lines of code - which solves said problem for any input. We then reason about upper time and space bounds for such an algorithm in terms of asymptotical complexity with regard to the input size in bits (or some general unbounded parameter which describes the difficulty of the input). When we are unable to find an algorithm suitable to our desires, we can reason about the constraints to which such - if it exists - must conform, in terms of lower-bounds. Furthermore, we can go on and analyze the relation between the relative hardness of problems - even of those for which we do not yet have a satisfactory solution - by clustering them into complexity classes and then proceeding to examine the relationships between these. In 2005 there were about 417 complexity classes in the Complexity Zoo \cite{cit01}. As of 2019, the number has grown to about 544 classes currently being investigated by humanity. However, they all pertain to solving some hard problem correctly and efficiently for all inputs, no matter how large.

A large number of problems, including SAT and all NP-Complete ones, integer factorization, solving stochastic games, all PSPACE-Complete problems and all problems in EXPTIME and above, as well as some mysterious ones like breaking AES encryption do not have any known efficient algorithm despite decades of research. These add to problems which are known not to admit any algorithm at all to solve them in the general case, including Kolmogorov complexity and the Busy Beaver Game. Some of the others are suspected to not to admit such, but wheatear this is actually true or not is yet unknown.

Nevertheless, a large number of instances of many of such problems are actually solvable in practice. Modern SAT Solvers \cite{cit02} can solve problems over up to millions of variables and a large number of those over tens of thousands \cite{cit03}\cite{cit04}. There is no presently known method of deriving an instance of a SAT problem which is hard to solve by any heuristic (although finding easy cases of arbitrary size can be done). In fact, not even a theoretical framework exists to reason about such cases, despite numerous published empirical studies. The incomputable Busy Beaver problem itself has been solved for the two symbol game, up to 4 states inclusively \cite{cit05}.

There is an apparent discrepancy between how hard a problem seems to be in the general case (at least for the human mind) and how easy it is for at least some practical cases, which are the ones of actual concern to us. As such, it is time to turn our attention to studying actual instances of cases of hard problems - in particular those which might appear in practice (and are thus almost always of bounded input size). Investigating these might sometimes lead to efficient algorithms for all real-world needs, or even to the discovery of the general case solution.

Until now, computer science theory has paid little to no attention to finite algorithmics. The very foundational tool used to reason about algorithms - asymptotic complexity of a function, works by definition only for the limit to infinity. Under existing theory, all practical instances of a problem - which almost always entail some bound on input size - are trivially solvable in $O(1)$. This includes computation of incomputable functions, without inclusion of proofs of correctness. 

In this paper we remedy this lacking of the current complexity theory by introducing finite algorithmics. We also present some noteworthy elementary results formulated under it.

\subsection{Prior work}
As far as we know there is no theoretical prior work concerning finite algorithmics, as we are just now introducing the field. Results exists which can, in retrospect, be regarded as part of field of finite algorithmics however. They can be found in the following domains: Artificial Intelligence and Machine Learning (there most problems solved have bounded input size by formulation), Heuristic Solvers for NP-Complete Problems (such as SAT Solvers) and to some extent Cryptography (since most cyphers have fixed key and block sizes). Nevertheless, even within these fields there has been to the best of our knowledge no systematic effort to date to introduce a theory which would allow formal reasoning about relative performance of various algorithms and relationships between various classes of problems with regard to a fixed upper bound on input size. The study of the P/poly complexity class can be considered tangential to this work. We will discuss its relationship to some of the other complexity classes we introduce.  

\subsection{Overview of this paper}
The rest of this paper is organized as follows.

In Section 3 we introduce the theory related to Finite Algorithmics as follows. Section 3.1 contains basic definitions pertaining to formulating computer science problems and solutions on the finite case. In Section 3.2 we introduce definitions which allow us to reason about natural functions restricted to a finite domain using concepts analogous to those employed in general case asymptotic theory. In Section 3.3 we present seven finite case complexity classes and define a few related concepts important to describing inherent difficulty within computer science problems. Section 3.4 formally introduces the problem of solving a computer science problem (i.e. producing the source code of an acceptable algorithm) and describes its inputs and outputs. Also in Section 3.4 we introduce a classification of existing general case computer science problems based on their known or apparent difficulty.

In Section 4 we present some elementary but very important results related to finite algorithmics, formulated within the theory we introduced in Section 3. Section 4.1 deals with relationships between finite complexity classes, both in relation to general case complexity (Sections 4.1.1 and 4.1.2) and among themselves (Section 4.1.3). In Section 4.2 we present a generic method for solving any computer science problem on the finite case (Section 4.2.1) and also introduce some very important elementary results pertaining to what performance guarantees can be attained for sure for certain types of problems (Section 4.2.2). In Section 4.4 we present further ideas which can be employed in the context of finite algorithmics to speed-up the quest for a solution to three well-known hard problems: 3CNF-SAT, String Compression (Kolmogorov Complexity) and Integer Factorization (hard only for a classical computer).

We use Section 4.3 to present 10 arguments which we consider overwhelmingly convincing in proving the existence of value in the study of finite algorithmics.

In Section 5 we discuss some clear implications of the results presented in this paper, including on the way we, as computer science researchers, ought to think about hard problems like P=NP.

In Section 6 we present four directions for immediate further research, and pose a crucial open question, within the realm of finite algorithmics. The answer to that open question can be used to decide (and for most practical purposes prove) P=NP. Furthermore, answering it can be done automatically (if but in a very long time frame). The mere existence of such a questions opens up new avenues in the quest for proofs in deciding P=NP.

Section 7 contains a brief Vitae of the author. Section 8 is dedicated to Acknowledgments and statement of interests (none).

Finally, In Annex 1 we include some estimated upper bounds for tractability for each finite complexity class, given existing hardware.

\section{Materials and Methods}\label{Sec2}

This paper contains results derived by theoretical reasoning based on the author’s current knowledge of advances in complexity theory, building of SAT Solvers and algorithms in general. Since it aims to introduce a new subfield of computer science, namely finite algorithmics, it stops short of providing experimental data as being out of the current scope. 

Obtaining such experimental data, based on the methods presented here is of interest nevertheless and we, the author, encourage fellow scientists to try them out in practice and publish the findings. 

Ultimately, the attractiveness of the field in general steams partially from the prospect of being able to enhance one’s creativity using computers to automate trial-and-error. They can perform tasks such as eliminating obviously unpromising alternatives several orders of magnitude faster than a human.

\section{Theory}\label{Sec3}

We now proceed to introduce the required theory which enables formal reasoning about finite cases of general computer science problems.

\subsection{Introducing Finite Algorithmics}\label{Sec31}
\begin{mydef}\label{Def1}\textbf{Problem of Restricted Size.}

Consider some problem $Prob$ consisting of finding a proper algorithm $S$ which, for any given an input $s$ of length $|s|$ from the universe possible inputs $U \in \{0,1\}^{*}$, produces some output $S(s)$ which is among the set of valid outputs for the given input. We define $Prob[n0]$ as the problem of finding such an algorithm which produces desired output only when $|s| \leq n0$. Such an algorithm can have undefined behavior elsewhere.\qeds
\end{mydef}

\textit{Example}: $SUBSUM[1000]$ is the problem of finding an algorithm which computes correctly whether a particular sum is attainable by summing some or all of at most $n0=1000$ given integers.

\textit{Discussion}: Note that the algorithm which is the answer to $Prob[n0]$ can be different for different $n0$-s. $SUBSUM[1000]$ might have a different algorithm than $SUBSUM[1000000]$. Also, solving the original problem $Prob$ entails providing an algorithm which solves it for any input, regardless of the size - the same for all sizes. Thus, $Prob[n0]$ can be regarded as a 1-parameter function $f:N \rightarrow \{a,b,c,...\}^{*}$, which, given some $n0$, outputs a string representing the desired algorithm in some chosen programming language. $Prob$ itself can be regarded as a parameter-less function (or a constant) providing such. 

\begin{mydef}\label{Def2}\textbf{Problem of Exact Size.}

We define $Prob(n0)$ analogously to $Prob[n0]$ to represent the problem of finding a proper algorithm when the input size is precisely $n0$.\qeds
\end{mydef}

We extend the notations of Definitions 1 and 2 to parameterized complexity accordingly. Namely, when we reason about the complexity of some algorithm not in terms of its input size, but in terms of some parameter $n$ (for example number of variables in a 3CNF-SAT problem instance) - which only bounds the input size but is not exactly equal to it, the same notations apply replacing the length $|s|$ of the input with the definition of this parameter. 

The definition of what constitutes a proper algorithm for a given problem merits attention. For a particular family of computer science problems (e.g. boolean formula satisfiability) a myriad of constraints can be placed on either inputs (e.g. no more than 3 clauses per variable), outputs (e.g. a satisfiable assignment if one exists should also be provided), the algorithm itself (should be no longer than 10 Mbytes) or its runtime behavior (e.g. space and time complexity), in addition to the type of machine which will be running it (e.g. a probabilistic computer, quantum computer) in order to arrive at a particularization which is specific enough to allow us to reason about it formally. Some constraints are more interesting than others though.

\begin{mydef}\label{Def3}\textbf{Full Problem Statement.}

In order to specify the statement of a computer science problem fully, we require the following to be included:
\begin{enumerate}[(1)]
\item \textbf{Theoretical problem statement.} This is a formal description which specifies which particular outputs can be considered correct for a certain input. \textit{Example}: For discrete logarithm we can consider an output correct if it represents the actual discrete logarithm of the input.
\item \textbf{Type of machine used to solve it.} This can be Turing-equivalent, Probabilistic Turing-equivalent or Quantum Turing-equivalent. If humanity discovers other types of machines, this list can be expanded accordingly, without losing validity of most results within this paper.
\item \textbf{Restriction on input size.} This can be specified directly, or via some parameter which constraints it. For the non-finite case, this limit is taken to be $+\infty$. We require that this limit either be $+\infty$ or a natural number explicitly given (not merely constrained).
\item \textbf{Restriction on output size.} This involves setting some constrains on the function which correlates the output of the algorithm to size of the corresponding input. \textit{Example}: We require output be of polynomial size in the input size.
\item \textbf{Restrictions on input universe.} In addition to size restrictions, we can require that the input satisfies some additional constraints, limiting generality (e.g. there be only 3 clauses per variable for a 3CNF-SAT instance, or that it represents a satisfiable formula). These can be included in (1) or not. 
\item \textbf{Accuracy requirements.} These specify how often and in what way is the algorithm allowed to stray from the strict correlation relationship between inputs and outputs defined in (1). For a decision problem, these can be acceptable rates of false-positives and false-negatives over all valid input pairs. They can be specified in absolute terms (i.e. a natural number), or as a bound on the fraction of such to some other quantity - for example constraining Sensitivity and Specificity. For non-decision problems, constrains on absolute or relative error can be included here. Finally, sometimes different requirements for different subsets of the input universe can be formulated (e.g. in case a 3CNF-SAT formula has less than 2 clauses per variable, we require 100\% Sensitivity and Specificity, but if it has more than that we can settle for 99\%).
\item \textbf{Proof Requirements.} This specifies if the algorithm must provide some sort of additional output which can be used to construct a proof that it is indeed correct for the respective input. For a 3CNF-SAT formula this can be a satisfiable assignment, or a certificate of non-satisfiability (do note for this particular example that not all non-satisfiable 3CNF-SAT formulas may have non-satisfiability certificates of polynomial size). We call any such part of the output a certificate (of correctness). Accuracy requirements can be placed on this part of the output as well.
\item \textbf{Completeness Requirements.} This specifies what kind of behavior the algorithm is guaranteed to have over the input universe. In particular, we say that it is \textbf{complete} if it terminates with the required guarantees for all inputs and \textbf{incomplete} if it does not do so for some of them (for which it may produce invalid outputs or simply never terminate).
\item \textbf{Restrictions on size of algorithm.} For some fixed programming language considered, we require that the size of the algorithm produced to solve the problem be bounded from above by some function of the input size. For a general case algorithm, this size must be a constant (however it may be rather large). For a problem of restricted size, it can vary with the input size restriction. Nevertheless, for a particular input size it must have a definite upper bound.\qeds
\end{enumerate}
\end{mydef}

We have deliberately excluded running time and space complexity of the algorithm from the problem definition. This is because for a given problem we will reason about its difficulty in terms of the running-time required to solve it. As with classical complexity theory, this can be taken for the Worst-Case, Average Case, Best Case or anything in-between (including ``average case in practice"). The memory model we employ is generally the \textbf{RAM model}. We typically do not include any mention of space-complexity, since by employing a Perfect Hashing scheme on the accessed memory addresses, space can be bounded from above by the time consumed, with only a small factor increase in the latter. We also typically, but not always, constrain the output to be of polynomial size in the input. We take space bounds to mean additional space besides that used by code of the program itself (which can be modified at runtime if needed!). Similarly we can exclude the time required to load the body of the program into memory (even if it may be extremely large - for example exponential in input size). 

Other machine-specific runtime requirements (such as number of random bits used or number of qubits employed) can be imposed accordingly as in the general case.

Proof requirements are specifically important when we are reasoning about algorithms we either do not know in advance, or about which we do not have sufficient insight to prove that they produce correct outputs for all inputs. For example, for determining the $k$-th bit of Chaitin’s constant \cite{cit06}, for $k$ between $10^9$ and $10^9 + 10$, an algorithm which simply outputs ``1" for all inputs, might in fact be correct for all we know. However, without some insight into why it is correct, this may not be satisfactory enough.

Note that a proof need not always be a requirement. Many image recognition and other algorithms constructed via Machine Learning provide no proof of the correctness of their outputs. In fact, for such algorithms, we currently more or less have little-to-no idea about both \textbf{why} they work so well in practice, and \textbf{when} they work this well (this latter failing has been shown to allow attacks for example against a road-sign recognition algorithm, which produce an image which to a human looks like a clear ``STOP" sign, but to the algorithm it is seems a clear ``Minimum speed 120 Km/h" sign). This has nevertheless not curtailed their adoption in practice.
Also note that the proof part of the output may include only what is required to complete or generate some larger proof (of potentially much larger size, e.g. exponentially larger) in some format which can convince either a human or, respectively, an automated proof verifier for the problem domain that the output is indeed correct. For a 3CNF-SAT instance for example, a proof of unsatisfiability could be just a small subset of the input variables - small enough to allow exhaustive trial of all possible assignments - which, when the input expression is reduced accordingly it generates empty (impossible) clauses.

The restrictions on the size of the algorithm itself are a novelty specific to finite algorithmics. For the general case, the implicit assumption made by humans in their quest for a solution is that there is a single algorithm (of some fixed size) which solves all inputs properly. The interestingness of our theory and of this paper in general rests on the assumption that some problems admit different algorithms (of potentially different sizes) for different input sizes - and that some may not even admit an algorithm for the general case. 

In some cases it can be useful to ``break up" an algorithm (its source code) into a fixed part, which is the same for all inputs in the problem space (similarly to a fixed algorithm for the general case) and a variable part - the ``hint" - which may vary with input size.

\begin{mydef}\label{Def4}\textbf{Algorithms with Hints.}

We define the solution to some problem $Prob$ (of either general or restricted size), to consist of a fixed proper algorithm $S$ which takes as input both the instance $inst$ of the problem and some hint data $hint$ to produce its output, alongside a function $GEN$ which generates the hint for a particular input size $n$. The output for a particular problem instance, is thus $S(inst, GEN(|inst|))$.
We call $S$ a \textbf{hinted algorithm}.\qeds
\end{mydef}

\textit{Discussion}: The advantage of having the $GEN$ function split from the rest of the algorithm's body is that it could be precomputed (note that it takes as parameter the size of the input, not the input itself). By taking $S$ to include a source-code interpreter (a machine simulator) and $GEN(n)$ to include some source code, we can describe any algorithm in this fashion.

For general case problems, if we constrain $GEN(n)$ to be polynomial in size to $n$, and $S$ to run in polynomial time, the algorithms examined will all be contained within the complexity class P/poly. Do note that problems which do not admit P/poly algorithms in the general case (e.g. the hint would grow to super-polynomial size beyond a certain threshold) might very well be solvable efficiently for all sizes involved in practice - up to potentially very large ones. Also, P/poly solutions for the general case may be of no practical use for some problems. Determining the hint may take exponential time, may be no less hard than the original problem itself or the P/poly solution may imply no constructive method at all to generate the hint or even determine if a particular hint is adequate. Alternative algorithms requiring much shorter hints in practice might exist, but they might not behave well for arbitrary large inputs thus placing the general case problem outside P/poly. Finite algorithmics can therefore be considered a field tangentially related to, but fully distinct from study of any general case complexity class, including P/poly.

\subsection{Finite complexity and its classes}\label{Sec32}

In order to be able to reason easily about relative running times of various algorithms, on the finite case - where regular complexity theory will simply give $O(1)$ - we would like to introduce some additional theory.

The easiest extension to the definition of asymptotic growth rate of some natural function $f:N \rightarrow N$ is to simply introduce an upper bound on the constant hidden by the $O$, $o$, or $\Omega$ notations. 

In the following we take a natural function to mean any monotonically non-decreasing function from natural numbers to natural numbers. This includes any function which can represent running time or space complexity of some algorithm for any input up to a certain size (difficulty).

\begin{mydef}\label{Def5}\textbf{Finite complexity with bounded constant and restricted domain.}

For two natural functions $f$ and $g$, some constant natural number $c$, and two other natural numbers $n1$ and $n0$, with $n1 \leq n0$, we say that $f(n)=O_{n1..n0}[c](g(n))$ \textbf{iff} $f(n) \leq c*g(n)$ for all $n$ between $n1$ and $n0$ inclusively.
We extend the definition accordingly to allow for $n0$ to be $+\infty$. 
When $n1$ is the minimum possible value in the input universe, we can omit it and specify only $n2$.\qeds
\end{mydef}

The above definition allows us to describe relative performance of algorithms in some familiar way. For example, for the All-Pairs-Shortest-Path problem, we can say that the complexity of the Floyd-Warshall algorithm \cite{cit07} is $T(n)=O_{+\infty}[100](n^3)$. This essentially means that all of the operations performed by this very short non-recursive algorithm (incrementing loop variables, dereferencing, comparisons and assignments) are no more than $100*n^3$ for a graph with $n$ vertices. This is definitely the case for any $n$ (there are probably less than 20 such operations per $n^3$).

The shortcomings of the above notation steam from the fact that, for finite cases, we have that $f(n)=O_{n1..n0}[c](g(n))$ for any two non-zero functions $f(n)$ and $g(n)$, for some appropriate constant $c$. Thus, we need clarify yet some more, for this approach to become useful.

The natural approach is to choose the constant as small as possible (introduce a tight bound).

\begin{mydef}\label{Def6}\textbf{Finite complexity with minimal constant and restricted domain.}

For two natural functions $f$ and $g$, and two natural numbers $n1$ and $n0$, with $n1 \leq n0$, we say that $Const_{n1..n0}(f,g) = c0$ \textbf{iff} (i) $f(n)=O_{n1..n0}[c](g(n))$ and (ii) $f(n) \neq O_{n1..n0}[c-1](g(n))$.
As before, we allow $n0$ to be $+\infty$ and $n1$ to be omitted where appropriate.\qeds
\end{mydef}

We are now able to reason about an algorithm $S$ in terms of ``if it were to have complexity $g(n)$, how large would the hidden constant need to be?"

\begin{mydef}\label{Def7}\textbf{Apparent relative finite complexity.}

For three natural functions $f$, $g$, and $h$, and two natural numbers $n1$ and $n0$, with $n1 \leq n0$, we say that $f(n) = O_{n1..n0}^{h(n)}(g(n))$ \textbf{iff} $Const(f,g)_{n1..n0} \leq h(n0)*Const(f,g)_{n1..(n1+n0)/2}$. 
We say formally that for the interval $n1..n0$, the function $f$ appears to have complexity $g$, within a factor of $h$.
When the function $h$ is constant, we can write the constant directly.\qeds
\end{mydef}

\textit{Discussion}: We have essentially constrained that from the mid-point of the interval, to its endpoint the constant grows by a factor of at most $h(n0)$. 

For a general case algorithm of some complexity $g(n)$, we have that there exists some $n0$, beyond which its apparent finite case complexity is also $g(n)$ within a factor of $h(n)=const$. This follows directly from the fact in the general case, beyond a certain threshold, the constant remains fixed regardless of $n$.

\begin{mydef}\label{Def8}\textbf{Certain finite complexity.}

For two natural functions $f$, and $g$ and two natural numbers $n1$ and $n0$, with $n1 \leq n0$, we say that $f(n) = OC_{n1..n0}(g(n))$ \textbf{iff} $f(n) = O_{n1..n0}^{1+1/n^2}(g(n))$. \qeds
\end{mydef}

\textit{Discussion}: We have chosen $h(n)=1+1/n^2$, such that $\prod_{n \rightarrow +\infty}(h(n))$ is bounded (it is in fact $\approx 3.68$). This allows us to reason that if $f(n) = OC_{n0}(g(n)) \Rightarrow f(n) = OC_{2*n0}(g(n))$ for all $n0$ beyond a certain threshold, then $f(n) = O(g(n))$. Any $h(n)$ with bounded $\prod_{n \rightarrow +\infty}(h(n))$ can be used to replace our choice.

\begin{mydef}\label{Def9}\textbf{Polynomial rank of a finite complexity.}

For a natural function $f$ and two natural numbers $n1$ and $n0$, with $n1 \leq n0$, we say that $PolyRank_{n1..n0}(f) = k$, \textbf{iff} (i) $f(n) = O_{n1..n0}^{2}(n^{k-1})$ and (ii) $f(n) \neq O_{n1..n0}^{2}(n^{k-2})$.\qeds
\end{mydef}

\textit{Discussion}: We have chosen $h(n)=2$, such that $\prod_{n=n0..2^{k}*n0}(h(n))$ after $k$ doublings of $n0$ is no more than $2^{k+1}$ (it is precisely this actually). This allows us to reason that if $PolyRank_{n0}(f(n)) = k \Rightarrow PolyRank_{2*n0}(f(n)) = k$ for all $n0$ beyond a certain threshold, then $f(n) = O(n^{k})$. 

\begin{mydef}\label{Def10}\textbf{Logarithmic rank of a finite complexity.}

For a natural function $f$ and two natural numbers $n1$ and $n0$, with $n1 \leq n0$, we say that $LogRank_{n1..n0}(f) = k$, \textbf{iff} (i) $f(n) = OC_{n1..n0}(log^k(n))$ and (ii) $f(n) \neq OC_{n1..n0}(log^{k-1}(n))$.
We consider only $k \geq 1$. If no such $k$ exists, we say that $LogRank_{n1..n0}(f) = +\infty$.\qeds
\end{mydef}

\begin{mydef}\label{Def11}\textbf{Linear finite complexity class.}

For a natural function $f$ and two natural numbers $n1$ and $n0$, with $n1 \leq n0$, we say that $f \in Linear_{n1..n0}$, \textbf{iff} $f(n) = OC_{n1..n0}(n)$.\qeds
\end{mydef}

\begin{mydef}\label{Def12}\textbf{Polylogarithmic finite complexity class.}

For a natural function $f$ and two natural numbers $n1$ and $n0$, with $n1 \leq n0$, we say that $f \in PolyLog_{n1..n0}$, \textbf{iff} $LogRank_{n1..n0}(f) \leq log(n)/log(log(n))$.\qeds
\end{mydef}

\textit{Discussion}: The value $log(n)/log(log(n))$ was chosen such that the resulting effective growth rate is linear or below.

\begin{mydef}\label{Def13}\textbf{Polynomial finite complexity class.}

For a natural function $f$ and two natural numbers $n1$ and $n0$, with $n1 \leq n0$, we say that $f \in Poly_{n1..n0}$, \textbf{iff} (i) $PolyRank_{n1..n0}(f) \leq 1+log(log(n))$ and (ii) $f \notin PolyLog_{n1..n0}$.\qeds
\end{mydef}

\textit{Discussion}: The value $log(log(n))$ was chosen such that any problem within this complexity class is most likely tractable for almost all inputs which show up in practice. For example, if we take $f(n)$ to represent the complexity of some algorithm based on its input size, for an input of size $2^{64}$ ($\approx$16 thousand Petabytes), then the maximum $PolyRank(f)$ for $f \in Poly$ is 7. Also for an input of mere $1024$ size, the maximum $PolyRank$ can still be $\approx$5. We consider this appropriate since a few interesting problems, like Assignment, have general case complexity around these thresholds. If the practical cases for the problem at hand involve $n \ll 1024$, the constant 1 in $1+log(log(n))$ can be increased to something more suitable, like 2 or 5.

\begin{mydef}\label{Def14}\textbf{Semi-Polynomial finite complexity class.}

For a natural function $f$ and two natural numbers $n1$ and $n0$, with $n1 \leq n0$, we say that $f \in SemiPoly_{n1..n0}$, \textbf{iff} (i) $PolyRank_{n1..n0}(f) \leq 1+log(n)$ and (ii) $f \notin Poly_{n1..n0}$ and $f \notin PolyLog_{n1..n0}$.\qeds
\end{mydef}

\textit{Discussion}: The value $log(n)$ was chosen such that any problem within this complexity class would most likely be tractable for a significant number of inputs which show up in practice. For example, if we take $f(n)$ to represent the complexity of some algorithm based on its input size, for an input of size $1024$, the maximum $PolyRank(f)$ for $f \in SemiPoly$ is 11, placing the problem at the threshold of tractability versus intractability given existing super-computers. Again, if in practice $n \ll 1024$, the constant 1 in the $1+log(n)$ can be adjusted to something more suitable. 

\begin{mydef}\label{Def15}\textbf{Exponential rank of a finite complexity}

For a natural function $f$ and two natural numbers $n1$ and $n0$, with $n1 \leq n0$, we say that $ExpRank_{n1..n0}(f) = 1/k$, \textbf{iff} (i) $f(n) = OC_{n1..n0}(2^{n/k})$ and (ii) $f(n) \neq OC_{n1..n0}(2^{n/(k+1)})$.\qeds
\end{mydef}

\textit{Discussion}: We are thus describing for a certain $n$, how large the exponent of 2 needs to be, in order to tightly provide an upper bound for the function, as a fraction of $n$ itself.

\begin{mydef}\label{Def16}\textbf{Exponential finite complexity class.}

For a natural function $f$ and two natural numbers $n1$ and $n0$, with $n1 \leq n0$, we say that $f \in Exp_{n1..n0}$, \textbf{iff} (i) $ExpRank_{n1..n0}(f) \leq 8$ and (ii) $f \notin SemiPoly_{n1..n0}$, $f \notin Poly_{n1..n0}$ and $f \notin PolyLog_{n1..n0}$.\qeds
\end{mydef}

\textit{Discussion}: We are taking the exponential finite complexity class to represent everything which is at most about simply exponential in $n$, which does not belong to better class. This is a break from the general case EXPTIME complexity class, where the exponent is allowed to be polynomial in $n$, not just linear. We have chosen the value 8 instead of 1, to allow functions of the order of $n!$ to fit into this class, up to $n \approx 512$. This should be more than enough for anything beyond it to be considered intractable in practice.

\begin{mydef}\label{Def17}\textbf{Intractable finite complexity class.}

For a natural function $f$ and two natural numbers $n1$ and $n0$, with $n1 \leq n0$, we say that $f \in Intr_{n1..n0}$, \textbf{iff} (i) $ExpRank_{n1..n0}(f) > 8$ and (ii) $f \notin SemiPoly_{n1..n0}$, $f \notin Poly_{n1..n0}$ and $f \notin PolyLog_{n1..n0}$.\qeds
\end{mydef}

\textit{Discussion}: We are basically naming everything above exponential finite complexity class to be Intractable. In practice, for some small $n<110$ (for example $n<20$), problems in this class may still be solvable. Nevertheless, if $n$ is this small, then the output for all possible inputs can be precomputed and given as a hint to a hinted algorithm.

\begin{mydef}\label{Def18}\textbf{Constant finite complexity class.}

For a natural function $f$ and two natural numbers $n1$ and $n0$, with $n1 \leq n0$, we say that $f \in Const_{n1..n0}$, \textbf{iff} (i) $f = OC_{n1..n0}(c0)$, for some fixed constant $c0$ and (ii) $LogRank_{n1..n0}(f) \leq 1$.\qeds
\end{mydef}

\textit{Discussion}: Constant finite complexity is quite similar to general case constant complexity. Do note however that the constant rank obtained in practice, might in fact be hiding some small growing non-constant function for the general case. Furthermore, when reasoning about complexity with regard to different upper bounds $n0$, the constant $c0$ must remain fixed - independent of $n0$.

\subsection{Finite complexity hierarchy}\label{Sec33}

Given the definitions in Section 3.2, for any given natural numbers interval $n1..n0$, with $ n1 \leq n0$, we can classify all the natural functions $f$, into precisely one of the following classes.
\begin{enumerate}[(1)]
\item $Const_{n1..n0}$
\item $PolyLog_{n1..n0}$
\item $Linear_{n1..n0}$
\item $Poly_{n1..n0}$
\item $SemiPoly_{n1..n0}$
\item $Exp_{n1..n0}$
\item $Intr_{n1..n0}$
\end{enumerate}

The higher the level a function occupies in this hierarchy, the less tractable we expect a problem admitting an algorithm of this complexity to be.

We are now armed with the ability to describe the variation in the classification of a particular natural function $f$, as we allow the input domain to expand. 

\begin{mydef}\label{Def19}\textbf{Threshold of complexity class explosion.}

For a natural function $f$, a complexity hierarchy level $l$ and a natural number $n1$, we say that $Explode_{n1}(f,l) = Min \{z \geq n1 | (\forall n0<z \Rightarrow \exists r \leq l, f \in Complexity_{n1..n0}(r)) \land (\forall r \leq l \Rightarrow f \notin Complexity_{n1..z}(r))\}$, where $Complexity_{n1..n0}(l)$ denotes the corresponding finite complexity class at level $l$ in the hierarchy. If the set is empty, we take the marker value $+\infty$. We can represent the hierarchy level by either its index or by the corresponding name (PolyLog,Linear,etc.). \qeds
\end{mydef}
\textit{Discussion}: We are taking the explosion threshold for a function $f$, to be the minimum $n0$ no lesser than some $n1$ value, where the function $f$ will belong to a finite complexity class strictly above the respective level.

Do note that the function $f$ might have different $Explode_{n1}(f,l)$ values, for different $n1$-s. Also, for some level $l$, there could exist some $n1$ beyond which the explosion threshold is $+\infty$. 

\begin{mydef}\label{Def20}\textbf{Threshold of complexity class collapse.}

For a natural function $f$, a complexity hierarchy level $l$ and a natural number $n1$, we say that $Collapse_{n1}(f,l) = Min \{z \leq n1 | \forall n0, n1 \geq n0 \geq z \Rightarrow \exists r \leq l, f \in Complexity_{1..n0}(r)\}$, where $Complexity_{1..n0}(l)$ denotes the corresponding finite complexity class at level $l$ in the hierarchy. If the set is empty, we take the marker value $+\infty$. We can represent the hierarchy level by either its index or by the corresponding name (PolyLog,Linear,etc.). We also allow $n1$ to be $+\infty$. \qeds
\end{mydef}
\textit{Discussion}: We are taking the collapse threshold for a function $f$, to be the minimum $n0$ beyond which $f$ belongs to a certain complexity class or better, at least for up to another higher limit $n1$ (which we allow to be $+\infty$).

\subsection{Finite Algorithmics and Problems}\label{Sec34}

When attempting to solve a computer science problem, we shall consider the following as input:
\begin{enumerate}[(1)]
\item The Full Problem Statement according to Definition 3.
\item The interval $n1..n0$ of input size (or other difficulty constraining parameter) where practical instances of the problem lie.
\item The worst acceptable finite complexity class for running time required by desired algorithm. We can reason in terms of worst-case/best-case/average-case either for the entire domain or simply for the instances which occur in practice. We can describe this by requiring that the running time belongs to a class up to a certain level of the finite complexity hierarchy described in Section 3.3. In practice this bound results directly from the following: point (2) above, point (4) below and the speed of existing hardware. If the produced algorithm allows high degree of parallelism, then the intended cluster size can also be factored in. See Appendix 1 for approximations considering currently existing hardware.
\item The amount of time which can be allotted to actually discovering the solution. When reasoning about a potentially variable upper input bound $n0$, this can also be expressed in terms of finite complexity class, with regard to the difficulty parameter $n0$.
\item Some collection of source-code for known algorithms and data structures.
\item The verifiability thresholds given existing algorithms. Namely:
\begin{itemize}
\item v1: The answers for ALL instances of this input size (difficulty) or below can be precomputed efficiently.
\item v2: The answer to ANY instance of this input size (difficulty) or below can be determined efficiently.
\item v3: The answer to MANY instances of this input size (difficulty) or below can be determined efficiently.
\item v4: The answer to SOME instances of this input size (difficulty) or below can be determined efficiently. For instances beyond this input size (difficulty), it is considered highly unlikely for then-existing state of the art to be able to compute an answer efficiently.
\end{itemize}
\item For instances of input size (difficulty) up to $v4$, some already existing correct input/output golden data can be available. Including it can be useful to save running-time during testing, by avoiding the need to run the original algorithm which generated it (which even though efficient, might still have consumed a lot of time or resources initially). Golden data can include:
\begin{itemize}
\item Precise outputs for some inputs.
\item Lower and/or upper bounds on the correct output for some inputs.
\end{itemize}
\end{enumerate}

The desired output for will consist of one of the following:
\begin{enumerate}[(1)]
\item A hinted algorithm $S$ and its fixed hint $hint_{n0}$.
\item A hinted algorithm $S$ and another algorithm $GEN$ which can generate $hint_n$ for any $n$ in $n1..n0$. We call $GEN$ the Hint Genesis Algorithm.
\item A hinted Algorithm $A_{n0}$ which generates the pair of algorithms from point 2 above, alongside its fixed hint, $hint_{A_{n0}}$. We call such an algorithm the Generator Algorithm.
\end{enumerate}

The most straightforward formulation of the above is the following: ``Given what we already know, find a sufficiently efficient, potentially hinted, algorithm which solves the full problem statement on any input of size (difficulty) within the interval $n1..n0$, or show how such can be constructed."

It is simple to note that an output of type (2) can be precomputed from one of type (3), by running the algorithm $A_{n0}$. Furthermore an output of type (1) can be precomputed from one of type (2). It is useful however to reason about these options separately, since the precomputation step which converts one into the other may not always be efficient.

When reasoning about the relative efficiency of finite case algorithms, we shall consider them both in terms of running time complexity and complexity of hint size, with relation to input size (difficulty). Thus, we say that an algorithm is $T(n)/G(n)$ efficient, where $T(n)$ is its running time and $G(n)$ its hint size. The hint, as well as the program code, are assumed to be already loaded in memory. We can reason about a certain algorithm / problem saying for example it has finite case complexity $Poly/PolyLog$ on the domain of interest.

The tables in Appendix 1 detail the maximum estimated tractable difficulty for each of the finite complexity classes, given currently existing hardware.

The computer science problems of interest to researchers can be classified into the following categories based what is currently known about the hardness of solving them in the general case.
\begin{itemize}
\item \textbf{Efficiently solvable (ES).} Problems in this category include string pattern-matching, shortest paths in graphs and many, many others. In fact most of the problems humanity has tackled are now included in this category. The state of the art algorithms known to the scientific community are sufficiently efficient to solve all practical instances of such problems.
\item \textbf{Tractable but insufficiently so (TR).} For some problems, like Assignment and Multidimensional Range Queries we know sufficiently efficient algorithms to solve any instance of them relatively quickly, but for some practical applications we need even faster ones. We may not know if such algorithms exist, as the gap between known lower-bounds and the upper-bounds can be quite large.
\item \textbf{Intractable for large input sizes, but tractable for small ones (PTR).} For problems such as Prime Factorization, Discrete Logarithm, Boolean Formula Satisfiability, Knapsack problem and others, an algorithm for solving them precisely is known, but the best one is still very inefficient (largely in terms of running time), thus making it suitable only for small input sizes. Some problems in this category (including some NP-Complete ones), might fall in the TR category for some practical applications, when a sufficiently accurate approximation algorithm is known, when the practical input sizes are small, or when the practical instances have some other trait (known or unknown) making them easier than the general case. For example, having a small target sum for the knapsack problem, or having a small number of clauses per variable for 3-CNF-SAT makes these belong to TR.
\item \textbf{Intractable because of assumed hardness (ITRA).} For problems in this category, no algorithm is known which solves any instance but those of trivial size and, furthermore, it is strongly suspected that none exists, because they belong to a certain complexity class. Problems such as Quantified Boolean Satisfiability which belong to the complexity class PSPACE-Complete are believed not to be solvable in polynomial time. However it is not known if this is so or not. Also, \#P-Complete problems like \#SAT are also believed to be harder than NP-Complete ones. But this is yet again also unknown.
\item \textbf{Intractable and mysterious (ITRM).} There are problems - like determining the encryption key used to encrypt a known plain text using AES given the cypher output - which are not known to belong to a specific presumably hard complexity class. Nevertheless, they are generally regarded to be intractable by mere fact that a large number of researchers have spent a lot of time thinking about them and yet no efficient algorithm has been published.
\item \textbf{Truly Intractable (ITR).} Some problems, like Halting Problem, Busy Beaver, Kolmogorov Complexity (very useful in compression and encryption), word problem for semi-Thue systems, determining the bits of Chaitin’s constant to non-trivial precision and many others have been proven to be incomputable in the general case. That is, no algorithm at all exists which solves them. This nevertheless does not imply they are incomputable for bounded-size input. Two-symbol Busy Beaver game for example has been solved precisely for up to 4 states \cite{cit05}.
\end{itemize}

Finite Algorithmics aims to provide efficient algorithms for problems in TR, PTR, ITRA, ITRM and ITR classes but only for the finite cases which occur in practice, without necessarily solving or giving a definite negative answer with regard to a solution for the general case. Furthermore, known algorithms for general-case problems (in any tractability category) can be used in the automated or semi-automated quest for efficient ones for the finite case.

\begin{mydef}\label{Def20a}\textbf{The Problem of Solving a Problem.}

Given a particular finite case computer science problem $Prob[n0]$, specified by the inputs 1,3-7 in this section (excluding the actual finite limits), we refer to the problem of solving this problem, as the problem of indentifying some algorithm $A$, takes as input a natural number $n \leq n0$, and generates the source code of some hinted algorithm, along with its hint, which solves $Prob[n]$. In case we are interested in solving the problem for any upper input size (difficulty) bound, we can take the $n0$ and $n$ to be $+\infty$ and allow $A$ to take this special value for its single parameter. We reason about the complexity of solving a problem in terms of complexity of the corresponding algorithm $A$. For some problem $Prob$ we use $Complexity(Prob, n0)$ to denote the finite complexity class of the problem of solving $Prob[n0]$.\qeds
\end{mydef}
\textit{Discussion}: The complexity of solving a problem can be thought of, essentially, as the running time of an algorithm which runs on some machine (e.g. a classical computer) which produces the source code required to solve any instance of such problem, up to some upper input (difficulty) bound, which itself is no larger than some n0. Note that solving a computer science problem is in itself a computer science problem, to which we can apply the entire theoretical framework presented.

\section{Results}\label{Sec4}

We now proceed to present some elementary results of high importance derived within the theoretical framework introduced in Section 3.

\subsection{Relationships between complexity classes, finite and general}\label{Sec41}

In this subsection we present basic relationships between finite case and general case complexity classes for natural functions.

\subsubsection{From finite case complexity to general case}\label{Sec411}

\begin{mytheo}\label{Theo21}\textbf{When finite case complexity implies general case complexity.}

For any natural function f the following statements are true:
\begin{enumerate}[(1)]
\item If we have that $f(n) = OC_{n0}(g(n))$ for some $n0$, and also that for any $n' \geq n0$,  $f(n) = OC_{n'}(g(n))  \Rightarrow f(n) = OC_{2*n'}(g(n))$, then $f(n) = O(g(n))$.
\item If we have that $PolyRank_{n0}(f) = k$, for some $n0$ and $k$, and also that for any $n' \geq n0$, $PolyRank_{n'}(f) \leq k \Rightarrow PolyRank_{2*n'}(f) \leq k$, then  $f(n)=O(n^k)$. 
When such fixed $n0$ and $k$ exist, we can say that $f$ belongs to the general case polynomial complexity class P.
\item If we have that $LogRank_{n0}(f) = k$, for some $n0$ and $k$, and also that for any $n' \geq n0$, $LogRank_{n'}(f) \leq k \Rightarrow LogRank_{2*n'}(f) \leq k$, then  $f(n)=O(log^k(n))$. 
\item If we have that $f \in PolyLog_{n0}$, for some $n0$, and also that for any $n' \geq n0$, $f \in PolyLog_{n'} \Rightarrow f \in PolyLog_{2*n'}$, then  $f(n)=O(n)$.
When such $n0$ exists, we can say that $f$ is grows at most Linearly. Depending on the exact PolyRank (if fixed), it may in fact grow just polylogarithmically. 
\item If we have that $f \in Linear_{n0}$, for some $n0$, and also that for any $n' \geq n0$, $f \in Linear_{n'} \Rightarrow f \in Linear_{2*n'}$, then  $f(n)=O(n)$.
Like above, when such $n0$ exists, we can say that $f$ is grows at most Linearly. 
\item If we have that $f \in Poly_{n0}$, for some $n0$, and also that for any $n' \geq n0$, $f \in Poly_{n'} \Rightarrow f \in Poly_{2*n'}$, then  $f(n)=O(n^{1+log(log(n))})$.
This is strictly speaking superpolynomial, but barely so. Also, it is sub-exponential. 
\item If we have that $f \in SemiPoly_{n0}$, for some $n0$, and also that for any $n' \geq n0$, $f \in SemiPoly_{n'} \Rightarrow f \in SemiPoly_{2*n'}$, then  $f(n)=O(n^{1+log(n)})$.
This is superpolynomial, but also sub-exponential.
\item If we have that $f \in Exp_{n0}$, for some $n0$, and also that for any $n' \geq n0$, $f \in Exp_{n'} \Rightarrow f \in Exp_{2*n'}$, then  $f(n)=O(2^{8*n})$.
When such $n0$ exists, we can say that $f$ belongs to the EXP complexity class.
\item If we have that $ConstRank_{n0}(f) = c0$, for some $n0$ and $c0$, and also that for any $n' \geq n0$, $ConstRank_{n'}(f) \leq c0 \Rightarrow ConstRank_{2*n'}(f) \leq c0$, then  $f(n)=O(1)$. 
When such $n0$ exists, we can say that $f$ is of constant growth rate.\qeds
\end{enumerate}
\end{mytheo}
\begin{myproof} The proof involves straight forward induction and computation of the bound on the product of the allowed constant growth rates under each corresponding definition. While the significance of this theorem is big, its proof is trivial enough to be omitted from this paper.\qeds
\end{myproof}
\begin{mycolo} Statements 1-9, remain true if the induction hypothesis of the second part is extended to refer not only to $n'$, but to all $n''$, with $n0 \leq n'' \leq n'$. \qeds
\end{mycolo}

\begin{mytheo}\label{Theo22}\textbf{When finite case complexity excludes general case complexity.}
For any natural function $f$ and any finite complexity level $l$, if there exists an infinite number of $n0$-s such that $Explode_{n0}(f,l) <+\infty$ then $f$ belongs to a general complexity class worse than the corresponding one for level $l$ given by Theorem 1. \qeds
\end{mytheo}
\begin{myproof} The proof is by contradiction, showing that no fixed constant can exists hidden by the $O$ notation for the corresponding general case complexity class.  We consider it rather trivial and omit it form this paper. \qeds
\end{myproof}
\textit{Discussion}: This theorem gives a direct criterion for excluding a favorable general case complexity for a function $f$, when we know that it behaves poorly in practice and we are also able to reason that there are infinitely many points where it continues to behave poorly. Do note we require that an infinite number of $n0$-s exist. It may be that the function $f$ has smaller finite complexity for any small enough finite interval. However once the upper bound of the interval is allowed to grow sufficiently large, the complexity always explodes.

\begin{mytheo}\label{Theo23}\textbf{Precise determination of general case complexity.}
For any natural function $f$ and any finite complexity level $l$, we have that $f$ belongs to the corresponding general case complexity class under Theorem 1, \textbf{iff} there exists an $n0$, such that $Explode_{n0}(f,l) = +\infty$. \qeds
\end{mytheo}
\begin{myproof} The proof consists of direct application of the definition of $Explode$ to reduce to a proper application of Theorem 1. We consider it trivial enough to be omitted from this paper. \qeds
\end{myproof}

The three theorems above provide for a method for converting between finite and general case complexity classes, when possible. Note that we may not always know enough about a function to be able to apply either of them. Also, even if a function belongs to a favorable general case complexity class, this does not mean a problem admiting an algorithm with such a running time is solvable in practice: its threshold $Collapse_{+\infty}(f,l)$ may be beyond the largest instances of practical interest. Analogously, even if a function belongs to a certain unfavorable general case complexity class or worse, the corresponding problem may still be very solvable for all cases of practical importance. The minimum $n1$ such that $Explode_{n1}(f,l) < +\infty$ may be larger than the maximum size of practical instances.

\subsubsection{From general case complexity to finite case}\label{Sec412}

\begin{mytheo}\label{Theo24}\textbf{General case complexity implies finite case after some threshold.}
For any natural function $f$ which has general complexity $O(g(n))$ and $\Omega(h(n))$, the following statements are true:
\begin{enumerate}[(1)]
\item There exists some $n0$, such that for all $n'$, $n' \geq n0$, we have $f(n) = OC_{n'}(g(n))$.
\item For any finite complexity class level $l$ with growth rate no smaller than $g(n)$, we have that $Collapse_{+\infty}(f,l) < +\infty$ and also that there exists some $n0$, such that $Explode_{n0}(f,l) = +\infty$.
\item For any finite complexity class level $l$ with growth larger than $h(n)$, we have that for any $n0$, $Explode_{n0}(f,l) < +\infty$. 
\end{enumerate}
By growth rate of functions in a finite complexity class $l$, we refer to the asymptotic growth rate of the formula representative of such class, given by the Definitions in Section 3.2. \qeds
\end{mytheo}
\begin{myproof} The proof is by contradiction, following the direct application of the definitions of general asymptotic growth notations. We consider the proof trivial enough to be omitted. \qeds
\end{myproof}

\textit{Discussion}: This theorem states that having determined some bounds for the general case complexity, these bounds will characterize the finite case as well after some threshold.

\subsubsection{Relationships between finite case complexity classes}\label{Sec413}

Relationships between finite complexity classes of different problems are trickier to characterize than in the general case. This is because in case of reduction from one problem to a number of applications of some others, it actually matters precisely how many applications there are of each of those others are used and also what the input size (difficulty) bounds for those are. As such, for example a polynomial number of applications of a solution of complexity class $Poly_{n0}$ might very well result in the $PolyRank$ of the main algorithm to exceed the $1+log(log(n))$ upper-threshold for the $Poly_{n0}$ class. Nevertheless, for a particular candidate algorithm we can definitely characterize its actual complexity with regard to the problems to which the solution is reduced, using the $LogRank$, $PolyRank$ and $ExpRank$ functions.

\begin{mytheo}\label{Theo25}\textbf{Reductions between finite case problems.}
The following statements are true::
\begin{enumerate}[(1)]
\item Up to $n^k$ applications of an algorithm of $PolyRank_{n0} = j$ results in an algorithm of complexity $PolyRank_{n0} \leq k+j$.
\item Up to $log^k(n)k$ applications of an algorithm of $LogRank_{n0} = j$ results in an algorithm of complexity $LogRank_{n0} \leq k+j$.
\item Up to $2^{n/k}$ applications of an algorithm of $ExpRank_{n0} = 1/j$ results in an algorithm of complexity $ExpRank_{n0} \leq 1/k+1/j$. \qeds
\end{enumerate}
\end{mytheo}
\begin{myproof} Again we omit proofs as they are mere algebraic symbolic multiplications of the formulas corresponding to the definitions of the respective classes. \qeds
\end{myproof}

There are of course more interesting reduction theorems. The core aspect for reductions within some finite interval $n1..n0$ is for the result to allow the candidate algorithm to fit the definition of a certain complexity class. This generally involves computing the new $LogRank$, $PolyRank$ or $ExpRank$ values. Theorems describing relationships between lower and higher finite complexity classes are also interesting. We leave such questions for further research.

\begin{myobs}\label{Obs26}\textbf{Variable finite complexity.}
For any complexity class level $l$, and any (potentially infinite) sequence of increasing natural numbers $a_1,a_2,...$ there exists a natural function $f$, such that for every $i$, we have $Explode_{a_i}(f,l)=a_{i+1}$. \qeds
\end{myobs}

\textit{Discussion}: It can be that a function has infinitely many points where its finite complexity is small enough, however it never collapses permanently to such a favorable case. A function which grows extremely fast on successive powers of 2, but very slowly in-between is one such example. In practice, it can be that a problem's optimal complexity varies wildly from one input size (difficulty) to another, within the bounds of some general-case complexity (if such bounds exists).

Finally, when solving a problem on the finite case, consulting the data in Annex 1 which estimates the highest upper bounds for tractability for various finte complexity classes can prove of general interest.

\subsection{Automated and assisted solving of computer science problems on the finite case}\label{Sec42}

In this subsection we present basic results regarding finite complexity classes of computer science problems. In Section 3.2 we introduced the concept of finite complexity for natural functions. Here we apply it to refer to functions which describe the running times of algorithms and respectively the sizes of hints.

Since Sections 4.1.1-4.1.3 refer to relationships between finite and general case complexity of natural functions in general, the results there apply to such representing the running time of an algorithm as well as representing the size of its hint. We adopt the notation $f(n)/g(n)$ from general complexity classes to reason about finite complexities analogously (e.g. $P/Poly$ translates to $Poly_{n0}/Poly_{n0}$) . When solving a problem for the finite case interval $n1..n0$, we consider the actual program itself, $S$, to be of constant size, as per Definition 4. As argued before, a program of fixed size exists which solves all intervals $n1..n0$ for some potentially distinct and perhaps very large $Hint_{n0}$-s: the one which includes an universal machine simulator.

\subsubsection{General Approach}\label{Sec421}

The manner in which we choose to split the actual algorithm for a finite case problem between the fixed part and the hint is as much art as it is science. Knowledge of the problem domain as well as trial and error may lead to various choices in this regard. Nevertheless, the value of finite algorithmics lies in the conjecture that some problems do not admit a sufficiently efficient algorithm for the general case (or that identifying such is not possible), but do in fact admit some (maybe distinct) algorithms for any finite universe of inputs.

In our quest to identify a practical solution to a finite case problem on interval $n1..n0$, or determine that such does not exists, we can take the following approach.

\begin{myapp}\label{App25}\textbf{Automatic Solving of a Problem.}
Given some problem $Prob$, specified by inputs 1-7 described in Section 3.4, we construct an algorithm to solve it by taking the following rough steps:
\begin{enumerate}[(1)]
\item \textbf{Consider some enumerable (potentially finite) family $F$ of hinted algorithms $S$}. This family can be specific to the problem domain of $Prob$: it can essentially describe ``what we can expect the source code of some algorithm which solves it to look like". Each algorithm in this family is hinted, as per Definition 4.
\item \textbf{Maintains some internal state $s$, describing the status of the search for a solution}. This state can be of rather large size, so long as it fits the space complexity bounds imposed on the automatic solving algorithm itself. It may consist of the following:
\begin{enumerate}[(a)]
\item Promising Algorithms and methods of Hint generation.
\item Algorithms and Hints which are adequate for some specific input sizes (difficulties).
\item Statistics on promising algorithms, hints and hint generation methods collected in step (3)(e) below.
\item Information analogous to points (a)-(c) but with regard to inadequate algorithms / hints / hint generation methods.
\end{enumerate}
\item \textbf{While a solution is not found, explore some more, by taking the following steps}:
\begin{enumerate}[(a)]
\item Pick some algorithm $S$ from the family $F$, based on the internal state $s$.
\item Choose a potential hint, $hint_S$ for $S$, from the finite set of potential hints, by some method $GEN_S$.
\item Evaluate $S$ preliminarily by running it on several new inputs within the relevant domain, using hint $hint_S$. For example, S can be run on Golden Data tests, for increasing input sizes (difficulties) from $v1$ up to $v4$. 
\item If $S$ takes longer than some upper bound for running time on a particular input, or its course of execution seems unpromising, halt it (to be potentially resumed later). Note that given the theoretical framework concerning finite complexity classes, for a given input size (difficulty) any desired complexity translates to a precise upper bound on the running time (i.e. number of operations of the algorithm). The constant is never ``hidden" in finite algorithmics.
\item Collect the following data with regard to the execution of algorithm $S$ with $hint_S$ on the relevant test cases:
\begin{enumerate}[(i)]
\item Running Time and Space Consumed for some input.
\item Statistics regarding Input/Output correlations. These can include error rates, such as Specificity and Sensitivity.
\item Full or partial snapshots of its internal memory state at various runtime moments.
\end{enumerate}
\item Use the data collected above to update the internal state $s$. Naturally, the data could be summarized or aggregated before or after updating the internal state, in order to reduce its volume.
\end{enumerate}
\item \textbf {When a sufficiently adequate pair of algorithm $S$ and hint generation method $GEN_S$ are identified, output them}. This entails outputing the respective source codes. \qeds
\end{enumerate}
\end{myapp}

\textit{Discussion}: The approach essentially considers algorithms and hints in some arbitrary order, tests them on available inputs and finally chooses the first one which is adequate enough. The art of producing an efficient implementation of this approach for some problem domain lies particularly in identifying some good manner in which to choose this arbitrary order.

By introducing Approach 1 we are effectively shifting attention of researchers from focusing solely on understanding correlations between input/correct output pairs within a problem domain, to focusing on correlations between structure and hints of algorithms to their relative performance / adequacy with regard to that domain. This approach can be reminiscent of the field of AI and Machine Learning.

For Step 1 of Approach 1 above, the following family of algorithms can be considered.

\begin{myapp}\label{App26}\textbf{Choosing the family of algorithms for a problem domain.}
Given some object-oriented programming language grammar (e.g. C\#,Typescript, etc.), consider only source codes (algorithms) which satisfy the following conditions:
\begin{enumerate}[(1)]
\item They include as reference some or all types corresponding to data structures and algorithms included in part 5 of the input, as described in Section 3.4. These can be restricted to only what researchers believe to be relevant to the problem domain.
\item They define at most 20 new types.
\item Each defined type contains most 20 public methods.
\item Each defined type contains at most 20 private methods.
\item Each defined type contains at most 20 internal variables. They can be collections such as lists or dictionaries over other types.
\item Each defined type contains at most 20 ``magic constants", which are actual values of some type.
\item Each method takes at most 20 parameters, all typed.
\item Each method defines at most 20 local variables (for all levels of nesting).
\item Each executable line of code, consists of one of the following:
\begin{enumerate}[(a)]
\item A $statement$, which can be either of:
\begin{enumerate}[(i)]
\item An $assignment$ to some variable in scope from the result of the evaluation of an $expression$ over some of the variables in scope.
\item An $evaluation$ of an $expression$ over some of the variables in scope without storing the result.
\item A loop-related execution flow control directive, such as $break$ or $continue$.
\end{enumerate}
By $expression$ above we understand to include invocations of public member methods on some of object stored in a variable in scope. Also, building parameter tuples using operators such as comma are also included.
\item A conditional branching of the form \textbf{IF($expression$) then $codeblock$ else $codeblock$}.
\item A loop of the form \textbf{WHILE($true$) $codeblock$}. This allows for both initial and final loop condition checking, via inclusion of an appropriate \textbf{IF} statement.
\item A return directive of the form \textbf{RETURN($expression$)}. 
\end{enumerate}
A $codeblock$ is understood to be a sequene of executable lines of code.
\item The expressions which occur throughout all methods are defined globally. Source code within methods uses them by specifying their corresponding index in the global list. The operators within an expression can be only names of methods on the underlying types, or the parameter list builder (e.g. comma). All types are boxed. As such, for example $a+b$ is represented as $Add(a,b)$, where $Add$ is a member method of the boxed number type.
\item There are at most 20 ``heavy" expressions globally, defined on between 6 and 20 variables.
\item There are at most 400 ``light" expressions globally, which are defied on at most 5 variables
\item Values are passed by reference to invoked methods. Thus, their actual value can be changed by the method if it so choses. Non-destruction of the input can be achieved via cloning.
\item No method has maximum nesting level above 5 (i.e. \textbf{IF} contained within another \textbf{IF} contained within a \textbf{WHILE}, and so on). Some algorithms with higher maximum nesting level can be rewritten to use private method invocations.
\item There are at most 20 ``heavy" methods globally, which have an imbrication level of either 4 or 5. All the others have imbrication level at most 3.
\item There are at most 400 executable lines of code for any single method body.
\item There are at most 20000 lines of code for the entire algorithm.
\item There are at most 400 global magic constants, besides the ones allotted to each type. 
\item There exists a single type, which is the actual Solver for the problem at hand which defines the following interface:
\begin{enumerate}[(a)]
\item Initialize(hint) - a method which initializes the solver with the corresponding hint, allowing any required precomputations to be performed.
\item Query(instance) - a method which returns correct output for a specified instance of the problem. \qeds
\end{enumerate}
\end{enumerate}
\end{myapp}
\textit{Discussion}: The above family of fixed algorithms is understood by us, the author, to include essentially everything that a human researcher could reasonably discover on his own with regard to any problem domain. In fact, to the best of our knowledge, almost all (if not all) currently known algorithms for solving any general case problem can be mapped to some member of this family. The value 20 which appears repeatedly was chosen arbitrarily to provide a rather generous upper bound. It is most likely that a lower value, like 5 will still allow sufficient coverage of everything which could be of interest to researchers. Furthermore, for specific problem domains, just a fixed algorithm of much smaller sophistication can be hypothesized to solve the problem on sufficiently large input sizes (difficulties), given the appropriate hint. In that case, the quest for a solution simplifies to the quest for a proper hint.

Note that the family of algorithms described in Approach 2 is universal. Namely, it includes algorithms which simulate a registry machine. Given that the algorithms are hinted, it is possible to effectively circumvent any limits placed on nesting depth, source code size or anything of the like, simply by moving the actual algorithm to the hint. This runs somewhat contrary to the manner in which we suggest that this approach has value. The hint should be specific to the problem domain. For some problems (like 3CNF-SAT for example) it could include some bending of this rule, like allowing for formulas specific to a particular input size to be evaluated in the context of conditional branching. For some very hard problems, it could be that allowing the program to modify itself by essentially incorporating parts of the hint or of the problem instance in its de-facto code might prove interesting avenues for exploration. Such problems might include simulation of human consciousness, resolving the halting problem for large instances and perhaps others not yet considered by humanity. Nevertheless, such flexibility should most likely not be the initial main focus of research within finite algorithmics.

Once a candidate fixed algorithm $S$ which shows sufficient promise on small problem instances has been identified - for example one which works correctly and efficiently on all inputs from a large test battery - , the only remaining issue is to determine a suitable hint for it for instances of desired size. The following approaches can be taken either alone or in combination for doing this:
\begin{itemize}
\item \textbf{Exhaustive Hint enumeration}. If a convenient finite complexity class is suspected for the hint size with relation to input size, simply enumerate all possible hints. This could also be the starting point for problems where we are essentially clueless as to what a proper hint might be. For small enough instances, the Deterministic Finite Cover Automata \cite{cit08} which correctly recognizes the finite language of correct input/outputs pairs for all such can be used as a hint by a linear time algorithm.
\item \textbf{Inductive Hint construction}. Identify an algorithm which, given some hints for problem instances of smaller input size (difficulty), it constructs one for instances of larger size. This algorithm itself could be sought automatically using Approaches 1 and 2. We suggest that it takes itself a hint of very small constant size (if any at all). The input on which it operates is the set of hints for smaller input sizes.
\item \textbf{Adaptive Hint construction}. Use existing and future techniques to determine causal correlations between events - such as Deep Learning - to analyze the data collected in step (3)(e) of Approach 1. This includes correlations between parts of the memory state at runtime and ultimate behavior of the algorithm - correctness, running time, and so on. Use the results of such analysis to prioritize some hints over others, as well as to eliminate obviously or apparently unpromising ones. These techniques can also be used to alter and combine successful hints so that the search for an adequate one converges faster on an acceptable solution.
\item \textbf{Tapping Randomness}. Include randomness in decision making with regard to which variation to try next or how to prioritize approaches. Many surprisingly efficient SAT Solvers today employ it.
\item \textbf{Multiple Arm Bandits}. Ultimately, the quest for a proper hint, using some automated method, involves allocating some amount of a finite resource - running time - to various existing or new avenues of exploration: be it an trying out an existing hint on more cases for gathering further data, transforming one by some rule, combining one with another under some yet another rule, or introducing some other random variation to one. Sometimes, the expected benefit of a particular method or transformation over another is unclear or cannot be known in advance. Taking such decisions, including with regard to how to balance exploration and exploitation pertains to a well-known computer science problem called Multiple Arm Bandits (see \cite{cit09} for example).
\end{itemize}

Finally, at all stages of the approaches described above, a human researcher could intervene and make adjustments based on his own creative and rigorous judgment, potentially leading to further speed-ups in the search for a solution.

Note the that the approach described in this subsection includes any currently known Machine-Learning algorithm, including Deep Learning: The output of the learning is in our terminology the Hint to the algorithm, which, itself is merely a simulator of a neural network. The actual learning algorithm (e.g. Reinforced Learning) is just one potential method to be used in line (3)(f) of Approach 1. Any such currently known learning algorithm is itself contained with the finite family of algorithms proposed by Approach 2.

Approach 1 could be refined to include the theory of Schmidhuber related to Gödel machines \cite{cit10}. This can be applied either with regard to proving correctness or to simply speed up the search for an optimal algorithm.

\subsubsection{Elementary Results}\label{Sec422}

In this subsection we present some elementary results pertaining to finite complexity of computer science problems, considering the approaches described in Section 4.2.1. 

We generally limit our attention to problems which have polynomial or smaller output sizes. This includes all decision problems (where the output is a single bit). The restriction that the output size be polynomial is, most of the time, natural. Outputs of super-polynomial size would, in themselves, require a long running time to merely write out. 

\begin{myobs}\label{Obs27}\textbf{Reduction to decision problems.}
Any problem which has an output of polynomial size in the input can be reduced to a liner number of applications of a related decision problem. \qeds
\end{myobs}
\begin{myproof}
Consider the decision problem asking ``does there exist any correct output for this input instance, which is smaller than the natural number $x$?". By using binary search over the output space for a given input instance, one can, in a number of probes linear in the input size (logarithmic in the size of the output universe) determine some correct output using the solution to the decision problem above. \qeds
\end{myproof}

\begin{mytheo}\label{Theo28}\textbf{Every verifiable problem admits a $Poly_{n0}/Exp_{n0}$ finite case algorithm.}
For any problem $Prob$, which admits a general case verification algorithm $V$ able to decide for any $(input,output)$ pair if the output is correct for the given input, there exists a hinted algorithm $S$, such that for any finite case upper bound on input size $n0$, there exists an appropriate hint enabling $S$ to solve $Prob[n0]$ correctly for all inputs and run in time $Poly_{n0}$. \qeds
\end{mytheo}
\begin{myproof}
Given a verification algorithm, one can immediately construct an inefficient general case algorithm which produces a correct output for any given input: simply enumarate all potential outputs for an input and use the verification algorithm to pick a correct one. Given this, the correct output can be precomputed for any input of size up to $n0$. There are $2^{n0+1} - 1$ such potential inputs. The corresponding correct outputs could then be stored directly as hint. An algorithm which, given an instance in this $2^{n0+1}-1$ universe, simply looks up position where the correct output is stored in the hint, using binary search, takes $log(2^{n0+1}-1) < n0+1$ steps to identify such. It can then merely output the corresponding answer, which is $Poly_{n0}$, resulting in a total running time within the $Poly_{n0}$ finite complexity class. The hint size itself remains within $Exp_{n0}$, for reasonably large $n0$-s. \qeds
\end{myproof}
\begin{mycolo}
Any verifiable decision problem admits a $Linear_{n0}/Exp_{n0}$ finite case algorithm for any $n0$. \qeds
\end{mycolo}
\begin{myproof}
Since a decision problem has constant output size (namely 1 bit) only the linear time taken to identify the index for the correct output determines complexity. \qeds
\end{myproof}
\textit{Complexity}: For decision problems, the above approach involves merely $2*(2^{n0+1}-1)$ applications of the verification algorithm $V$. For problems of polynomial output size this is multiplied by the maximum size of the output universe, which is of the order $2^{O(n0^c)}$ for some constant $c$. In both cases, executing this approach directly for any $n0$ has general case complexity within EXPTIME, so long as the algorithm $V$ is itself within this class (e.g. it is in P). While not generally considered tractable, EXPTIME is not the worst general case complexity class out there.

\textit{Discussion}: The above approach gives an algorithm to determine some correct output, but not also a proof of its correctness. For the brute-force method, correctness is guaranteed by the construction of the hint itself: The output is correct, because given the manner in which the hint was constructed, it cannot be incorrect.

\begin{mytheo}\label{Theo29}\textbf{Solving verifiable problems optimally in the finite case is computable.}
For any problem $Prob$, which admits a general case verification algorithm $V$ able to decide for any $(input,output)$ pair if the output is correct for the given input, and for any natural number $n0$, there exists an unhinted algorithm which determines the optimal algorithm for solving $Prob[n0]$. \qeds
\end{mytheo}
\begin{myproof}
Any implementation of Approach 1, which exhaustively enumerates all algorithms and potential hints of joint size up to at most the size of the algorithm and hint resulting from the application of Theorem 6 for $Prob[n0]$ will consider the optimal running time algorithm among them. Hint sizes outside $Exp_{n0}$ are pointless, since an algorithm linear in input plus output size (thus optimal) exists for such a large sized hint. Considering algorithms which, together with their hints, are of size over $Exp_{n0}$ is again pointless: the algorithm constructed in the proof sketch of Theorem 6 is of very small constant size, thus allowing the joint size to remain within $Exp_{n0}$. As such, trial of algorithm/hint pairs only within these limits suffices. \qeds
\end{myproof}
\textit{Complexity}: The complexity of employing this approach without further refinement is as follows. For a decision problem, for every candidate algorithm/hint combination in the input universe, verifying its correctness can take at most $2^{n0+1}-1$ applications of the verification algorithm $V$ and similarly many applications of the candidate algorithm itself. These are within EXPTIME if $V$ and the candidate are within EXPTIME themselves. For problems with larger, but still polynomially sized outputs, this is multiplied by some factor $2^{O(n^c)}$ for some constant $c$, representing the increased output size. This keeps the complexity for veryfing a single algorithm/hint pair within EXPTIME, so long as $V$ and the candidate are themselves within EXPTIME. The universe of potential algorithm/hint pairs is doubly exponential in $n0$ (exponential in the maximum hint size), making the total complexity no worse than 2EXP, which is within ELEMENTARY. Large, but computable.

\textit{Discussion}: In practice, the trial universe will be much smaller. Most likely only hints of $Poly_{n0}$ or $SemiPoly_{n0}$ size will be considered and the family of fixed algorithms for a problem domain will consist of just 1 or sometimes a very small subset of those described by Approach 2. Also, most candidate algorithm/hint pairs will not be allowed to run beyond some desired complexity (most likely $Poly_{n0}$ or $SemiPoly_{n0}$) and will not be run on all possible inputs for verification purposes, thus resulting in further running time reductions.

\begin{mytheo}\label{Theo30}\textbf{Solving efficiently verifiable problems optimally in the general case is computable if they have a determinable complexity collapse threshold.}
For any problem $Prob$, which admits a general case verification algorithm $V$ able to decide for any $(input,output)$ pair if the output is correct for the given input, and also has a known or determinable $n1$ such that $Explode_{n1}(Prob,l)=+\infty$ for some desired complexity hierarchy level $l$, there exists an unhinted algorithm which provides a general case algorithm of complexity corresponding to finite complexity level $l$ which solves $Prob$, so long as the verification algorithm $V$ is belongs to this complexity level or better itself. \qeds
\end{mytheo}
\begin{myproof}
One can apply the method in Theorem 7 for ever increasing $n0$-s (for example taken under repeated doubling or repeated squaring) until it can be established that the $n1$ threshold has been reached. If an upper bound is known for it in advance, $n0$ can be taken to be directly $n1$. The method in Theorem 7 is modified to output not just one suitable algorithm, but all of them. This multiplies the size of the output of the method by at most the size of the universe of algorithm/hint pairs, making it as large as $2^{2^{O(n0)}}$. The correct algorithm, which solves the general case problem within the desired complexity, is necessarily amongst this outputted set. In terms of general complexity theory, this set has size $O(1)$. As such, the general case algorithm constructed consists of running all such algorithms ($O(1)$ of them) for any input instance given, and verifying each of their outputs using the algorithm $V$, picking the correct one. So long as the complexity of $V$ is no larger than the desired class, this results in an algorithm of such general case complexity. \qeds
\end{myproof}
\textit{Complexity}: The complexity of employing this approach without further refinement is essentially within at most some $log(n1)$ factor (for repeated doubling) of the complexity of a single application of the method in the proof sketch of Theorem 7 for $n1$, modified to output any acceptable algorithm/hint pair. The modification does not alter the running time complexity of this brute force method. As such, as argued for Theorem 7, the running time is within a constant factor of $2^{2^{O(n1)}}$, which, ironically enough is $O(1)$ in terms of general case complexity.

\textit{Discussion}: Note that by Theorem 4 any problem which admits a general case algorithm of some corresponding complexity (e.g. P) necessarily has some corresponding fixed $n1$. In practice, $n1$ may or may not be knowable in advance. It can be guessed or some rule for its determination hypothesized. For example, it can be speculated that if the finite complexity class has not exploded for 5 successive repeated squaring applications, then this threshold has been met or exceeded. Or it could be speculated that its finite complexity class is monotonically non-increasing with increase in input-size beyond some small fixed threshold.

Note that the output produced by the approach in Theorem 8 can be further trimmed down, both in practice (as some candidates are eliminated as more and more input instances are processed) and via theoretical reasoning, by a researcher which is able to prove that one such is actually always correct. A formal proof of this may itself be rather lengthy (e.g. consider the proof for classification of algebraic finite simple groups, which ``has around 15,000 pages, spread through mathematics literature"). If it exists at all! Given Gödel's first incompleteness theorem (see \cite{cit11}), there are true statements expressed in first-order logic over natural numbers which cannot be proven. The desired statement of correctness might happen to be one of them. In the eventuality a proof exists, a researcher could again employ the theoretical framework and approaches described in this paper to develop an algorithm to automatically find it. This is possible since verifying formal proofs is in fact rather straight forward - a trivial verification algorithm exists. Finding such a formal proof, or showing that one does not exists - that is the hard part. 

The significance of Theorems 6-8 is rather major: It shows that finite algorithms can, at least theoretically, solve almost all problems currently considered hard - if only, as of now, in 2EXP time - which may itself be a rather long wait. Nevertheless, their existence allows the problem of finding a solution to a computer science problem to be rephrased in terms of tradeoffs between the following three dimensions:
\begin{itemize}
\item Running Time of the solution algorithm $S$.
\item Hint Size for the solution algorithm $S$.
\item Running time of computing a suitable hint for algorithm $S$ by some automated method, given existing knowledge. 
\end{itemize}

There are undoubtedly many interesting questions and results pertaining to reductions and relationships between finite case problems and either general or other finite case problems. These include the ``meta-problems" induced by some general problem, under some interesting or practical assumptions: the problem of finding a suitable algorithm for it, of finding a suitable hint for such and others. Further interesting questions include applying the approaches here recursively upon themselves, to potentially produce faster than brute-force algorithms for solution finding. We leave such questions outside the scope of the current paper and propose them as avenues for further research.

One direction which seems particularly useful to us for priority examination consists of solving problems which fall in the following categories:
\begin{itemize}
\item They belong to the $Poly_{n0}/LogRank_{n0}=2$ finite complexity class. This entails the existence of a reasonably small number of potential hints (up to $n0^{log(n0)}$) - thus making exhaustive search relatively feasible.
\item They have known and relatively efficient output verification algorithms. All NP-Complete problems fall into this category.
\item They admit natural formulations as decision problems. Determining satisfiable assignments for a Boolean formula is one such example (with the 3CNF-SAT decision problem). Computing Discreet Logarithm is not.
\item There exists some non-trivial amount of solved hard test cases within existing body of research. 
\end{itemize}

\subsection{Reasons for considering finite algorithmics valuable}\label{Sec43}

The most important argument which needs to be made for acknowledging the importance of the study of finite algorithmics is why we should expect that finding a solution to the finite case of problems is easier than finding one for the general case.

We, the author, present the following arguments as indications that this is in fact the case:
\begin{enumerate}[(1)]
\item \textbf{There exist problems which are incomputable in the general case, but computable for any finite case}. Consider the following problems:
\begin{enumerate}[(a)]
\item For classical computers: \textbf{Optimal String Compression} - determining the shortest program which outputs a given string. In the general case this is equivalent to determining the Kolmogorov complexity of the string and is incomputable. However, if we limit the problem to the practical application of considering only strings of length up to some $n0$, and limit the running time of the program to at most $Exp_{n0}$, the problem becomes computable: one can simply enumerate all programs of length no larger than $n0$, and run them until they either time out, produce a wrong string or produce the desired string. Afterwards the shortest correct one can be selected. In fact, one could even try to determine a suitable hinted algorithm using Approach 1 and some hint generation method which allows some, or most strings which appear in practice to be compressed efficiently (NB: With regard to any algorithm, there are strings which are incompressible). The original proof by contradiction stating that Kolmogorov complexity is incomputable relies on the assumption that length of the source code of any such algorithm is shorter than the length of at least one target string. This however does not apply to the finite case, where strings have bounded size: A Kolmogorov complexity computation algorithm can indeed exists for this case (and does: for example one which includes precomputed output for each string), however it is necessarily longer than $n0$. Thus, solving the finite case, $Exp_{n0}$ bound running-time Kolmogorov complexity is actually in EXPTIME.
\item For Finite Automata: Recognizing Prime Numbers (PRIMES) - using a Finite Automata to determine if a string denoting the representation of an integer in some base (e.g. unary, binary, etc.) corresponds to a prime number. Finite Automata are unable to recognize this language. So for this computation model, PRIMES is incomputable. The proof is a relatively straight forward contradiction, using the Pumping Lemma \cite{cit12}. The problem of recognizing all prime numbers up to some upper threshold $n0$ however is computable using Finite Automata: it is in fact a finite language and all finite languages are regular. Note however, that different automata are required for different $n0$-s.
\end{enumerate}
This illustrates that the finite case can be simpler than the general case. In fact, it is so for some very practical problems: Kolmogorov complexity features prominently within Information Theory and Cryptology.
\item \textbf{There exist problems with large thresholds of complexity explosion}. Consider the following problem, inspired by the Theorem of Classification of Finite Simple Groups: ``Take some sort order for finite simple groups, such that groups of smaller order appear before groups of larger order. When tied, consider some other arbitrary criterion, like number of generators or anything else desired. Given an index $k$ of a group in this sort order, and a series of pairs of numbers representing elements within this group, output the result of the group operation acting on the elements, under some fixed (but arbitrary) numbering for them. The length of this series is logarithmic in the group order." The problem asks essentially to compute group operations consistently within a specified finite simple group. We can name it GROUPOP. Here the difficulty parameter is not the input size, but the order $n$ of the group which also bounds the input size. As per the Theorem of Classification of Finite Simple Groups, there exist actually only three infinite classes: cyclic groups of prime order, alternating groups of degree at least five and groups of Lie type. All of these have simple representations. However, there exist another 27 finite simple groups which do not belong to any of the infinite families. Out of these 27, the Monster Group, of order $M \approx 8*10^{53}$, stands out as it does not have a simple representation. As such, performing group operations in any of the other finite simple groups is computationally much faster than in the Monster Group. For cyclic groups for example doing group operations is as simple as multiplication modulo the prime which is the order of the group. This is actually logarithmic in the value of the group order. For the Monster Group however, Wilson has described a method involving two 196882x196882 matrices \cite{cit13}. Doing operations with these matrices is computationally very expensive, bringing GROUPOP outside Linear. Some other constructions have been proposed, however it still remains that the Monster Group is terribly difficult to work with. As such, one can say that $Explode(GROUPOP,Linear) = M \approx 8*10^{53}$. In fact, if operations within groups of the other two infinite families besides cyclic can be done in polylogarithmic time, we have that $Explode(GROUPOP,PolyLog) = M \approx 8*10^{53}$.

The value $8*10^{53}$ is rather large – large enough to be considered non-trivial. The problem GROUPOP is rather simple up to this group order, and then it explodes drastically. Could it not be that something similar happens to other interesting problems, like Integer Factorization or 3CNF-SAT? Furthermore interestingly, given the fact there is a single Monster Group, the complexity of GROUPOP ultimately collapses back: the super-linear complexity for $M \approx 8*10^{53}$ is ultimately smaller than a single $log$ factor of some larger group order. As such, we can state that $Collapse_{8*10^{53}}(GROUPOP, PolyLog) < +\infty$. Essentially, the finite complexity of GROUPOP is bitonic - small at first for quite some values, then it grows drastically large (rather quickly) and then collapses back to being small. From the point of view of a general case complexity theory, the existence of the Monster Group is fully irrelevant. The extreme difficulty of doing operations there, given the fact it is a singular finite case, is in fact $O(1)$. Thus, finite algorithmics offers a much better way to describe the structure of such problems.
\item \textbf{There exist problems where precomputation specific to a certain input size is very useful}. Consider the problem of determining the Minimum Spanning Tree for a given graph with $n$ vertexes and $m$ edges. This problem is relatively easy and numerous general case algorithms of near-but-not-exactly optimal complexity exist: from Kruskal's $O(m*log(n))$ to Chazelle's near-linear $O(m*\alpha(m,n))$ \cite{cit14}. However, there exists one algorithm by Petite and Ramachandran \cite{cit15} of optimal complexity - which, mysteriously enough is still unknown. Whatever it is, their solution is nevertheless bound by it. Their approach involves precomputation of all optimal decision trees on $log(log(log(n)))$ vertices. In this situation precomputation can be completed in $O(n)$, which is no larger than the complexity of the outstanding part of the algorithm. As such, the precomputation step can be done on-the-fly for each instance of the problem, without any need to store it as a hint to some hinted algorithm separately, for purposes of improving running time performance. 

It could be that some very difficult problems (maybe even 3CNF-SAT) have solutions which involve precomputations for an input size (difficulty) of larger complexity class than the rest of the algorithm. If the result of these precomputations is short enough however, we could simply store it as a hint to some hinted algorithm, belonging to a favorable complexity class such as $Poly_{n0}/Poly_{n0}$. Note that computing the hint for some input size could belong to a much larger complexity class – such as $SemiPoly_{n0}$ or $Exp_{n0}$. However, this only needs be done once for all inputs of that size. Once computed, if it is short enough it can be stored as hint allowing this potentially extremely time-consuming step to be skipped. We, the author, strongly suspect that if practical solutions for finite or general case 3CNF-SAT problems exist, they will involve reasonably sized hints which require nevertheless significant amounts of running time to compute.
\item \textbf{Other expressive computational models show significant drop in complexity from general case to finite case}. Consider potentially the closest relative of the Turing Machine - the Finite Automata. Consider a regular language with an infinite number of words. It can be succinctly described by a Deterministic Finite Automaton with some number of states. This number could then be reduced by computing the minimal automaton for the given language. So long as the language has infinitely many words, this is the best which can be achieved. However, when the attention is directed to a finite subset of this language - namely that of words which do not exceed some fixed finite length, it has been shown that the number of states could be reduced even further, using something called a Deterministic Finite Cover Automaton. This is basically an automaton which correctly recognizes the language up to words of at most the specified length, but it is allowed to error on anything longer. This is analogous to considering the finite case of some general case problem, where the sought-after solution is a specification for a Finite Automaton, not a registry machine. Deterministic finite cover automatons are expected to have a significantly smaller number of states than their counterparts for the unrestricted language. In fact, it has been shown that they have a smaller number of states than even their counterparts which recognize just the finite language precisely (are not allowed to error on longer words) \cite{cit08} \cite{cit16}.

It could be that classical computers exhibit a similar phenomenon for at least some languages - namely that complexity of recognizing one such up to some finite length is much smaller than that of recognizing it on the general case. While classical computers are a much stronger computational model than finite automata, the two are still closely related. For example, every bounded-time registry machine algorithm can be represented as an automaton which is initially fed the input and them some number occurrences of a special symbol, each corresponding to one clock tick of processing by the classical registry machine. The states of such an automaton are in fact the all the memory configurations the registry machine could encounter during its execution. While this representation is inefficient, it serves to illustrate the close relationship between the two computational models, for the finite case.

Even for machines of larger or incomparable computational power (e.g. Quantum Computers, or the theoretical Blum-Shub-Smale machines \cite{cit17}), the fact there exists sufficiently expressive computational models (the Finite Automatons) which experience complexity collapse for the finite case, serves as an indication the same could occur for these models also.
\item \textbf{There exist interesting problems which are outside 2EXP on the general case}. Reachability in Petri Nets \cite{cit18} is practically interesting and has recently been shown to be outside ELEMENTARY, thus outside 2EXP \cite{cit19}. Or consider any EXPSPACE-Hard problem for instance. Deciding if two regular expressions which allow squaring (requiring exactly two adjacent copies of the operand) represent different languages is in EXPSPACE \cite{cit20}, as is the validity problem for extended linear temporal logic with times. Besides these, many problems within Game Theory are PSPACE-Complete (e.g. solving generalized Tic-Tac-Toe), while others still are actually incomputable. 

For most of these categories of problems, there is no hope of solving them in practice by discovering an efficient algorithm for the general case. The only hope to ever solve these is within finite algorithmics - solving not the problem in general, but some restriction of it to a finite case. Here, one can apply Theorem 6 to show that a $Poly_{n0}/Exp_{n0}$ algorithm exists. Finding one however, may be outside 2EXP since the verification algorithm itself could be outside 2EXP. Nevertheless, the existence of a $Poly_{n0}$ algorithm (if but of exponential size) shows that the finite case is indeed easier than the general case for interesting practical problems. A prominent result within finite algorithmics will be one which gives a solution to one of these practically important, but generally intractable problems for some non-trivial practical upper bound.
\item \textbf{Finite case problems are a particularization of the corresponding general case problems}. Essentially, we as researchers have reduced the practical finite-case problems we are interested to solve to some potentially harder ones - namely their general case. While sometimes the general case is easy enough, this is not always true. As the easiness of 2-CNF-SAT relative to arbitrary Boolean formula satisfiability illustrates, sometimes the particularization is much easier than the generalization. Further research should focus on relationships between finite case and general case for specific problems, to determine where this is the case and where not. The theoretical framework introduced in Section 3 serves as a tool.
\item \textbf{There exists an automated method for finding an optimal solution to verifiable problems in the finite case}. The proof sketch of Theorem 7 shows how an optimal algorithm for such problems can be constructed. There exists an analogous result from Jones \cite{cit21} for general case verifiable decision problems. He essentially constructs an algorithm which runs, in a dove-tailing fashion, all conceivable algorithms until one stops and produces the correct output. While asymptotically this is optimal for the general case, the hidden constant is astronomical - it is exponential in the index of the suitable algorithm in the enumeration. This makes it generally unusable in practice. Note that Jones' method does not actually identify the suitable algorithm. For each problem instance, there could be some different algorithm which finishes first and outputs the correct answer (which is then verified by the verification algorithm). In the finite case, on the other hand, after spending some initial (potentially very large) amount of time, the optimal algorithm is determined (its source code becomes available). Thus, it can thereafter be directly applied to any instance (up to the finite upper bound) where it performs efficiently enough. Using Jones' method for the finite case would entail dealing with the astronomical constant on every run of the algorithm - on every instance. Furthermore if a problem did not admit a general case efficient algorithm, his method no longer yields an algorithm of optimal complexity, since the index of the most efficient algorithm is no longer a constant.
\item \textbf{There have been prior successful applications of the approach of automatically generating algorithms}. The field of AI and Machine Learning, particularly Neural Networks is a perfect example where trying out and adjusting an algorithm, within a certain family results in something very useful. For most AI and Machine Learning applications (such as image recognition), the problem researchers were trying to surmount was the apparent lack of a proper formal description relationship between input and output. For example, describing formally what ``an image of a cat" was (or to go further, what ``an image of a happy person" was) proved very difficult. Nevertheless, this was circumvented by employing an automated method of trial-and-error to essentially determine an algorithm which is good enough.

The same could be applied to the situation where the difficulty lies not in identifying a formalism to describe the input/output relationship, but in finding an efficient algorithm to compute it, if but only in practice. As with AI and Machine Learning, we can now regard this process as the result of a combination of automated trials and researcher insight, not just of the latter.
\end{enumerate}

The arguments above which serve to indicate that finite-case problems are indeed easier to solve (at least to us humans, potentially aided by computers) than their general case counterparts. However, there are two more arguments of a more abstract nature to indicate the existence of value in the approaches presented:
\begin{enumerate}[(1)]
\item \textbf{There exist problems which admit rather simple and short efficient algorithms, but which require complex theory to prove their adequacy}. The clearest example can be considered the string matching algorithm due to Knuth-Morris-Pratt (KMP) \cite{cit22}. It has less than 10 lines of code, a single method with no recursion, maximum nesting of 3 and all its expressions are over no more than 5 variables. Nevertheless, the sophistication of the theory behind it, especially with regard to proving its linear running time complexity, is the most likely cause why it has not been discovered earlier.
\item \textbf{Physical phenomena might exist which can be harnessed to allow rapid speedups for computations, if but at some great cost}. Given current mainstream understanding of physics concerning time dilation, if we were able to send a computer with sufficient battery power to a place far away from any gravity wells (like planets, stars or black holes) and have it stand as still as possible relative to Earth, some important speedups can be attained. Alternatively, we could leave such a computer on Earth and take refuge ourselves near a black hole until it finishes computing the desired output. Other phenomena might exist which to allow for much greater speedups (quantum non-locality seem like a good place to start a search). These however, might entail travelling to distant regions of the Cosmos, or expending large amounts of resources, like battery power. However, this can be regarded as a one-time-cost. With the methods provided for by finite algorithmics, such a sped up computer could then rely back the optimal algorithm (e.g. via radio waves) for some problem. Thereafter, we could use it solve all practical instances, without the need to incur the one-time-cost ever again. 
\end{enumerate}

\subsection{Application of techniques to three well known problems}\label{Sec44}

In this subsection we present some directions for practical application of the theory and techniques presented within this paper to three specific hard problems. They are intended to be viewed as just an example of how the quest for adequate solutions can be altered with the introduction of finite algorithmics. It is outside the scope of this paper to propose (much less test experimentally) a fully specified approach or method which can be employed to solve them. Nevertheless, we, the author, are confident that ideas formulated within the context of finite algorithmics - either based on the ones presented below or others - will eventually lead to an adequate solution to such, if one exists.

The main intended contribution of this section is to show, by way of example, how the change in reasoning due to finite algorithmics can lead to fundamentally different avenues of research for hard problems, from the ones currently pursued by computer scientists.

\subsubsection{3CNF-SAT}\label{Sec441}

The following ideas can be applied to solving 3CNF-SAT, in the context of finite algorithmics:
\begin{enumerate}[(1)]
\item \textbf{Consider only families of hard cases}. For example, for an $n$-variable formula, do not include in analysis clause configurations which are conjunctions of two or more formulas over less than $n$ variables, as such could be solved recursively separately. Also ignore families of known easy cases. For current heuristics published in literature this includes formulas with less than 2 or more than 5 clauses per variable.
\item \textbf{Discover specific problem structure incrementally}. Examine what makes some 20-variable 3CNF-SAT formulas harder to solve than others, for some algorithm or family of algorithms. It could be the existence of some tuple of clauses or some computable trait of a larger subset of clauses. Then look at 21-variable formulas and find which additional traits (besides those applicable from the 20-variable case) predict hardness. Then at 22-variable and so on. How many additional ``hard" formulas specific to an $n+1$-variable case are there (excluding those for formulas in up to $n$ variables)? Do they belong to some finite number of families? Is the number of such growing rapidly or slowly? How can they be described succinctly so as to potentially allow their storage as hint to some algorithm? Do same for 200- or 2000- variables randomly built 3CNF-SAT formulas. Such analysis could be aided by tools from AI and Machine Learning, which may discover unexpected or counter-intuitive correlations.
\item \textbf{Discover what makes candidate algorithms succed in solving hard instances when they eventually manage it}. For hard instances examine what actual choices (e.g. ``lucky random assignments") allowed their eventual resolving. How do these choices correlate with the input instance (or part thereof) and between themselves? Is this knowledge (or at least part thereof) common to several hard instances? Can all such knowledge for some $n$-variable instance difficulty be represented succinctly and efficiently enough so as to allow a $Poly_{n0}/Poly_{n0}$ algorithm to use it in order to solve all such much faster on subsequent runs? Is at least part of it common to many instances? Like with point (2) above, AI and Machine Learning tools might prove very valuable.
\item \textbf{Discover predictors for candidate algorithms non-performance}. Discover similar correlations as those in point (3) above, but for situations where a candidate algorithm performs very poorly. How can ``really poor" choices be described formally and succinctly, so they can be avoided on subsequent runs?
\item \textbf{Consider a family of algorithm/hint pairs (or a fixed algorithm with a family of hints). Discover correlations between parts of the content of the algorithm/hint themselves and performance/adequacy in trial runs}. What are good predictors for great/poor performance? Is it having a particular while loop in a certain place? Or doing random-restarts in some describable fashion? Researchers have been more or less attempting this step manually so far - leading to the discovery that random restarts are key to performance of advanced SAT Solvers \cite{cit03}. However, formal methods from AI and Machine Learning and not only could be employed to deduce many more such correlations much faster.
\item \textbf{Consider correlations between memory states of a candidate algorithm within a family and its performance/adequacy in trial runs}. It could be that for some family of algorithms, a certain memory state (or part thereof), if encountered at runtime, is strongly correlated with very poor performance (for example a particular choice of random assignments, or set of impossible assignments deduced). It can be regarded as similar to steps (3)-(5) above. Unlike steps (3) and (4), analysis aims to learn something not about the structure of instances themselves, but about structure of the family across some test battery, seeking to predict desirability of having the memory state characterized in a particular fashion. Unlike step (5), the analysis focuses not on the source code / hint contents used, but on the actual runtime memory state (which might be common, at least partially, to several algorithm/hint pairs). Like with points (2)-(5) above, tools from AI and Machine learning (and not only) can prove useful.
\item \textbf{Perform the same analysis as in point (6) above for a short sequence of memory states}. Essentially, this calls for analysis to be expanded from examining single snapshots of memory to examining short ``movies" of such (not necessarily sequentially chosen).
\item \textbf{Use some form of automatic recombination and selection method to generate new candidate algorithm/hint pairs and maintain the set under consideration within desired size limits}. This entails essentially using the information gathered in points (2)-(7) above to rank, modify and combine algorithms / hints such that one representing an adequate solution is found much quicker than by exhaustive enumeration. An example of a modification is to make an algorithm do a full or partial restart every time its runtime memory state can be characterized as ``unfavorable" as per data obtained under methods (6)-(7) above. An example of a combination is to run two or more algorithms in some dove-tailing fashion for a number of steps and then decide how to continue based on their joint memory state. Other methodologies like those specific to genetic algorithms or those presently employed in AI and Machine learning can be readily used. Ranking or more specifically selection is required to keep the candidate set size within the space limits imposed by whatever hardware is used in attempting to find the solution.
\item \textbf{Use ideas in points (1)-(8) above to incrementally generate algorithm/hint pairs for increasing difficulty}. This way, the information collected with regard to solutions of instances of lesser difficulty (smaller number of variables) can be exploited to speed up and obtain similar information about the solution to more difficult instances.
\end{enumerate}

The ideas described above are meant to speed up some automated or semi-automated search for a suitable algorithm. However, one such may not exist. Even if that is so, having a good choice of a heuristic algorithm, accompanied by a good choice of a hint can result in a huge drop in running time. It could be that such drop is much higher than the time actually consumed to generate the pair. After all, as per Corollary 2, 3CNF-SAT admits a $Linear_{n0}/Exp_{n0}$ algorithm. So the quest is actually for a more acceptable tradeoff between hint size (more specifically hint generation running time) and algorithm running time.

A final trick could be employed in practice. The universe of potentially hard formulas over $n$ variables is rather large. It is of the order of ${{4*n*(n-1)*(n-2)/6} \choose {4n}}*2^{4n}$ for formulas with up to 4 clauses per variable, which is much larger than $2^n$ and even than $2^{4n}$. However, in practice we might be interested in solving just a very small subset of these – namely the ones which occurred as a reduction from some other practical problem. Sometimes, it could be that we are actually interested in solving a single very lengthy formula - one for example giving a winning strategy for a complex military game position, or an optimal design for a microchip. In such a case we can particularize further. When using ideas (1)-(9) above (and any others for that matter), we can consider only expressions which are formed by a subset of the clauses appearing in the original large instance. Thus, the universe of instances for all variable sizes is cut to $\approx 2^{4n}$ which is a huge reduction. Furthermore, the ideas and methods described can now make use of structure specific to the original input instance to arrive at a solution much faster. The only draw-back is that the algorithm / hint pair can be expected to perform adequately only on the original input instance set (which may have a single element). This nevertheless, can be an acceptable and desirable tradeoff.

\subsubsection{Kolmogorov Complexity}\label{Sec442}

Problems which in the general case are provably incomputable due to reduction from Kolmogorov complexity typically relate to string compression or have to do with entropy extraction (generating more randomness from less such). The two are not unrelated. 

In this subsection we focus our attention on string compression. Formulating the problem for practical use in this case involves more than just restricting the input size. Formally we consider a string compression problem to have the following statement: ``Given a set of strings of length no more than $n0$, determine some pair of (potentially hinted) algorithms $Compress_{n0}$ - which takes an input string and produces a digest - and $Decompress_{n0}$ - which takes a digest and produces the original string -, such that the digest of maximum (or average) length is as short as possible (or simply ``short enough") and both algorithms belong to $Poly_{n0}/Poly_{n0}$ finite complexity class."

The above is an adaptation of the original formulation of Kolmogorov complexity, which required compressing a single string by using an algorithm of unbounded complexity (but which surely terminates) to produce it. The above formulation allows for the input set of strings to contain a single element as well. However, in practice it is more likely that a single solution is sought which can be used to compress several strings (potentially all the strings of length $n0$).

Under the above formulation, all ideas from Section 4.4.1 can be adapted here as well. The only difference will be in verification - as more and more algorithms are considered, performance entails not only examining running times but also lengths of generated digests.

Some ideas specific to string compression, formulated in the context of finite algorithmics are the following:
\begin{enumerate}[(1)]
\item \textbf{Determine short incompressible strings which appear as substrings within the input set}. It is a well-known information theory result that for any length there exist incompressible strings. This can be shown via a simple counting argument for a binary alphabet. Furthermore, the density of incompressible strings is rather large. Using this information, one can attempt to ``break down" the input strings into incompressible ``atoms". These can then serve as part of a hint to an algorithm which only describes how to assemble them together to obtain the desired string. Incompressibility within this context does not need to be strict. A reduction of less than 3-4 characters for example could make a large string be considered just as well incompressible. Note that determining atoms is incomputable in the general case for sufficiently large strings. Nevertheless it is very computable within the finite-case formulation we are considering.
\item \textbf{Given a list of short strings (atoms) determine a method which uses such to build a larger target string}. One straightforward such method is to break the target string into concatenations of atoms and then to store only the index of each such for each part. More sophisticated methods could involve exploiting correlations between contents at different positions (e.g. repeat adjacent occurrences of an atom).
\item \textbf{Apply ideas (1)-(2) above recursively, on the digests generated by the method in idea (2)}. This allows further compression based on the non-randomness of the pattern in which atoms themselves occur within a target string. Note that the input for the recursive step is typically strictly shorter than the original input, which was already compressed by a prior application. The final output could then just include a number indicating how many times recursion was applied.
\item \textbf{Consider space-time tradeoffs in deciding which short strings to keep as hint to the solution algorithm and how to represent them}. Atoms themselves may not need to all be kept in their lengthy, full form. While a single atom is considered incompressible, a list of several may have a more succinct representation than simple enumeration of them all. For a binary alphabet, a suffix tree or even a simple trie may offer an efficient improvement. However, there may be other shorter representations which in turn require longer processing times to allow the extraction of some ``$k$-th atom".
\item \textbf{Exploit randomness}. Consider producing methods and algorithms which make random choices. In the context of decompression, such can produce the desired original string only with some probability (e.g. $1/2$ or $2/3$) – and in the other cases either produce something else or exceed desired running time. In the context of compression, such could produce valid digests only with a certain probability. 
\item \textbf{Consider error-correction codes}. In the context of idea (5) above, consider padding some lengthy atoms using error correcting codes. While counterintuitive, this could potentially allow for shorter algorithms / hints to be generated - since one such need not output a precise string, but any of its correctable forms. Furthermore, simple detection of errors could indicate the need for rerunning said algorithm automatically for a different random seed, thus improving the probability of correct output. Finally, given some input set of strings, all atoms within it might be sufficiently separated in terms of Hamming distance. Thus, there may be no need for additional padding. An algorithm which only very occasionally outputs the correct atom and the rest of the time something which is not an atom can be combined with an algorithm (like a Deterministic Cover Automaton) which simply recognizes the language of atoms for the given input string set.
\end{enumerate}

All results pertaining to Kolmogorov extractors (entropy extractors), polynomial-time randomness (producing outputs which are indiscernible from random by any polynomial time algorithm) and related topics are relevant and can be further refined to apply in this context. A prominent researcher in this field is Prof. Marius Zimand (see \cite{cit23} or \cite{cit24}).

As illustrated in Idea (4) briefly, a related problem to string compression is finite language recognition: ``Given a set of strings, produce an algorithm which can determine if an input string is part of this set or not." This related problem is extremely relevant to finite algorithmics. Firstly, any decision problem can be formulated in terms of determining if an input instance is within the set of instances for which the answer is ``Yes". A solution to efficiently deciding membership within this set solves the original problem for the finite case. In fact, compression of the set of outputs of some problem (e.g. 3CNF-SAT) on some small finite input universe, such that set membership can be decided efficiently, can and should be employed in the course of running automated methods for finding the solution for larger input sizes (difficulties). The starting point for solving a decision problem can be finding the Deterministic Finite Cover Automata for the set of strings which represent small input instances with a ``Yes" answer. Using such, group membership can be decided very quickly (linearly in instance size). However the size of the hint (the actual description of the DFCA) can grow rather large. Nevertheless, we the author consider the relation between DFCAs and acceptable algorithms (in terms of running time / hint size / hint generation time) for set membership decision problems as a prime candidate for future research. We see such research as both general and specific to a particular problem domain (e.g. to the set of satisfiable 3CNF-SAT formulas over at most $n0$ variables).

\subsubsection{Integer Factorization}\label{Sec443}

Factoring large integers can be solved efficiently by quantum computers, using Shor's algorithm \cite{cit25}. Nevertheless, a similarly efficient algorithm for a classical computer is yet to be discovered. Integer factorization occurs mainly within the realm of cryptology and generally pertains to identifying a prime factor of a large semiprime number. Besides adaptation of the ideas from Section 4.4.1 which can prove useful, an idea specific to this problem is the following:
\begin{enumerate}[(1)]
\item \textbf{Identify and store ``hard" primes}. Given a target range for the integer to be factored (e.g. 512-bit or 1024-bit sized), and some state-of-the-art existent algorithm (e.g. Pollard's Rho algorithm or GNFS, or a combination of such), determine what constitutes ``hard primes" for it. These are prime numbers which, when they appear in the composition of an integer to factor, cause the algorithm's running time to increase drastically. If the number of such ``hard primes" is relatively small in relation to maximum value of the integer to factor, they could all be stored. Even if there are relatively many such, ideas from Section 4.4.2 could be employed to get a more succinct representation of this set, allowing it to be enumerated efficiently.
\end{enumerate}

The above idea, steams from the following anecdotal empirical experience of the author. Many years ago, he participated in an open factorization challenge (which was part of a larger computer science contest). Contestants were asked to factor each of 10 large numbers within a week. The author encountered the following situation: The first 7 were relatively easy to factor and he managed to factor the 8th and the 9th as well using some more advanced techniques. However, the 10th one seemed unbreakable. At that point we considered the following question: ``How could the problem settlers have come up with such a hard case in such short a time? It was known to him that they themselves had only about one week to prepare the challenge.” Given this, he tried the following: He searched on the internet for the primes which showed up as factors for the other two hard cases - namely the 8th and the 9th. He then identified a small number of short lists of primes which featured them. He then used a computer program to try out each of the primes on those lists against the hard 10th challenge case. To his delight, this worked. The ``hard prime" for the 10th case was in fact taken from a list on the internet. The experience above serves to indicate that generating ``hard primes" is no easy task. Like with 3CNF-SAT, most large instances of Integer Factorization are easy to solve. Those which remain may be hard due to the presence of some of these hard primes in the solution. Identifying all such and, if there are not that many, including them as hint to some hinted algorithm, might make integer factorization easy for all practical sizes even for a classical computer.

\section{Discussion}\label{Sec5}

We have discussed the significance and implications of most results and theory throughout the paper, close to the place of their introduction. In this section we present a few ideas of a more general significance.

The results in Section 4 serve to illustrate that analyzing a problem for the finite case, rather than on the sometimes more difficult general case holds value. Problems which are very hard (or even impossible) to solve in the general case may have acceptable finite case algorithms. Furthermore, the search for suitable algorithms in the finite case can be automated or sped up using computers.

The introduction of finite algorithmics allows us, as humans, to reason about hard problems differently. Ultimately, within the framework introduced in this paper one could ultimately prove that:
\begin{enumerate}[(1)]
\item \textbf{$P \neq NP$}. One way this could be done is by proving that for any large enough input size upper bound $n0$, the length of the shortest hint for a $Poly_{n0}$ time algorithm which solves it is strictly larger than for $n0/2$. This would not necessarily entail that NP-Complete problems cannot be solved in practice, but would decide the general case question.
\item \textbf{$P = NP$}.  One way this could be done is by providing a polynomial time algorithm which constructs a hint for any input size upper bound $n0$ for an algorithm of bounded $PolyRank$ time complexity. This could be further restricted to practical significance, by providing a $Poly_{n0}$ algorithm for hint construction for a $Poly_{n0}/Poly_{n0}$ algorithm.
\item \textbf{$P = NP$ or $P \neq NP$ but we really do not care about the distinction for practical purposes}. This could occur either because an efficient algorithm and hint have been identified for all practical bounds (favorable case) or because it has been proven that the shortest hint size for most practical cases is too large (unfavorable case). In the former situation, if $P \neq NP$ this essentially happens for input sizes outside of humanity's practical zone of interest, while in the latter, if $P=NP$ this again happens for too large input sizes, such that the drop in complexity in the general case is in fact of no practical use.
\end{enumerate}

The same discussion as above applies to the study of relationships between other complexity classes (such as between P and PSPACE).

The results and techniques presented in this paper can be applied not only to hard problems (PTR and above), but also to those which are relatively easy but for which we would like to identify even more efficient algorithms (TR). One such candidate is multidimensional range querying. An algorithm which breaks the ``curse of dimensionality" - if such exists - could be sought and found using the same approaches.

Ultimately, we expect the change in mindset and in focus of research resulting from rephrasing a problem in terms finite algorithmics theory to lead, in the near future, to practical solutions for some of the hardest computer science problems that have been haunting humanity for the past decades.

\section{Conclusion and Further Research}\label{Sec6}

Throughout this paper we have identified several avenues which we consider prime targets of future research. We briefly recap them here:
\begin{enumerate}[(1)]
\item \textbf{Examining relationships between different finite complexity classes}. This can pertain to relationships between different finite complexity classes for the same problem domain (e.g. for different $n0$ upper bounds) or between different problem domains (e.g. resulting from reduction of one problem to another). Also, they could be unspecific pointing out interesting results for finite complexity in terms of natural functions in general without the need for them to represent something in particular.
\item \textbf{Examining relationships between finite complexity classes and general case complexity}. Similarly this can occur within a problem domain, connect several problem domains or be unspecific, pertaining only to natural functions in general.
\end{enumerate}

With regard to the these, we ask simply ``What are interesting results which fall into these categories?". We presented a few elementary ones ourselves in this paper, in Section 4.1.

In addition to the above, we propose the following directions for future research, which seem to us promising: 
\begin{enumerate}[(1)]
\item \textbf{Investigating the relationships between Finite Automata and efficient Hinted Algorithms for the finite case}. Limiting input size, running time and usable space to some finite bound allows a problem to be solved within a computational model less powerful than a Turing machine. Namely, any algorithm on a classical computer which has bounded memory size and is limited to a maximum number of steps to perform (finite case complexity) can be accurately represented by a finite automaton over a ternary language: The states of the automaton represent the memory configurations which can be encountered during execution, transitions correspond to the small changes an algorithm can perform in one step leading from one memory configuration to another and the ternary language symbols represents the clock ticks which the algorithm consumes. The first part of an input word is the binary representation of the input instance for the original problem, and all the rest are 3-s. If the automaton accepts on such a constructed input, so does the corresponding classical computer algorithm. Ironically enough, not all automatons defined in this fashion correspond directly to a classical computer algorithm - a transition within an automaton can be from a corresponding memory state to any other, while for a classical computer a transition (one operation) only changes one word of memory at a time and thus it can point only to very similar states. While direct automaton construction and minimization based on the observation above may not lead to a time-wise feasible approach to solving a problem, conceptually it can offer deep insights. The relationship between the two computational models for the finite case warrants further research.
\item \textbf{For a specific problem domain investigate the growth of minimum hint size as the finite upper bound increases}. The fact a problem is limited to the finite case does mean the upper input size (difficulty) bound should remain fixed during analysis. While for practical applications existent at some moment such bound is a definite, effectively reaching it may entail examining correlations between solutions for smaller ones. One very interesting question is the following: ``Given a problem $Prob$ and some target finite complexity class for an efficient algorithm, how does the size of the shortest hint vary with the upper input size (difficulty) limit $n0$?" For general case, the answer is very simple: ``It is 0 for all cases". Finite algorithmics on the other hand allows further nuance.
\end{enumerate}

Finally we propose a specific, explicit question framed within the theory of finite algorithmics which, when answered, will give the strongest indication ever - if not a proof - for deciding the classical $P=NP$ problem.

Consider some fixed, sufficiently expressive hinted algorithm. Such an algorithm can simply be one which receives, as part of the hint, the index of a more elaborate algorithm from the finite family described in Approach 2 and then runs such on the remaining hint and input instance. The family in Approach 2 can be considered to include as ``predefined types" all popular data-structures and solvers for general case problems which are commonly known from literature as of November 2019.

Given the above fixed algorithm, answer the following question: \textbf{``What is the minimum length of some required hint, which allows the above algorithm to decide satisfiability for any 3CNF-SAT formula over at most $2^{20}$ variables within running time $Poly_{2^{20}}$?". Then answer the same question for $2^{30}$ and $2^{40}$}.

Firstly, answering these questions constructively will give the most efficient method possible for solving 3CNF-SAT in practice.

Secondly, by examining how the shortest hint size required grows for the $2^{20}$, $2^{30}$ and then for the $2^{40}$ upper bounds on number of variables, one can get the strongest indication - if not even a sufficient proof - with regard to whether $P=NP$. If the hint sizes increase (at least significantly), this is a very strong indication that $P \neq NP$. In fact, the only way this could happen and still have $P=NP$ is if the additional sophistication in the structure of the 3CNF-SAT with the increase in the number of variables, drops to 0 beyond a certain finite bound above ${2^{40}}$ (similarly to that of GROUPOP beyond the order of the Monster Group). We, the author, believe it to be extremely unlikely for 3CNF-SAT to behave so. Conversely, if hint sizes do not (at least not significantly) increase this would be a crushingly strong indication that $P=NP$.

Finally, if we were asked to take a guess, we would expect the answer to the above question to indicate a rather slow, but positive growth rate. Most likely on the order of $Poly_{n0}$. This would indicate that NP is outside P, however it would place it well within $Poly_{n0}$ or $SemiPoly_{n0}$ in practice. Furthermore, depending on how difficult computing such a hint proves to be in the general case, NP might be placed outside P but below EXPTIME.

We conclude this paper here.

\section{Vitae}\label{Sec7}

Mircea Digulescu is a computer scientist and software engineer. He was awarded bronze medal at CEOI 2004 as well as 4th and 10th positions at ACM SEERC 2005 and 2006 respectively. He is still active in competitive programing on Codeforces where he had reached the first division. He has obtained Bachelors and Masters Degrees in Computer Science at from University of Bucharest - Faculty of Mathematics and Computer Science, where he had also been studying as a PhD Candidate in applied computer science. His main interests are within Complexity and Computability Theory, Game Theory, Algorithms and Data Structures as well as Cryptology.

\section{Acknowledgments}\label{Sec8}

No organizations funded the research presented in this paper. The author's last affiliation is PhD candidate at the University of Bucharest, Department of Computer Science of Faculty of Mathematics and Computer Science. The author is currently an independent researcher. Statement of interests: none.

I would like to thank late researcher Mihai Patrascu for his lecture held at an a training camp for competitive programming contestants many years ago, where amongst other things, he revealed the existence of a deterministic linear time algorithm for solving the Minimum Spanning Tree problem, which worked only when the input size was greater than $10^{80}$. His remark that anything below this size was solvable in $O(1)$ served as an inspiration which ultimately contributed to the discovery of the ideas in this paper.

Warm thanks also to the three beautiful persons who inspired strong interest in solving hard computer science problems in practice. This paper would never have existed without them.

This paper is based on the ideas contained in the rough preprint published by the author in November 2019 \cite{cit26}.

\section*{Annex 1}

The tables below detail the maximum estimated tractable difficulty for the finite complexity classes. They assert 10 MFlop/s for a single-core on commodity hardware (from the author's empirical Codeforces.com experience), 83 TFlops/s for a single-core on super-computer grade hardware, a number of 2 million cores for the fastest super-computer and 60 million for all the TOP500 super-computers combined. 

Data was compiled directly from https://www.top500.org/lists/2019/06/. 

The values for multicore architectures (supercomputers) assume the candidate algorithm can be parallelized perfectly. Furthermore, these bounds are for a classical computer. Where random data is required, depending on its quality, generating one such word (or bit) may take longer than 1 Flop. Also, no similar bounds are provided for a quantum computer.

The bounds are given as follows:
\begin{itemize}
\item Single Core Commodity (SCC).
\item Single Core Super Computer (SCS).
\item Top Super Computer (MCT).
\item Top 500 Super Computers combined (MCA).
\end{itemize}

$\newline$
$
\begin{array}{|l|l|l|l|l|}
\hline
\bf{ExpRank = 1}&{SCC}&{SCS}&{MCT}&{MCA}\\
\hline
1 second&{\bf 16}&{\bf 32}&{\bf 46}&{\bf 50}\\
\hline
1 minute&{\bf 20}&{\bf 36}&{\bf 51}&{\bf 54}\\
\hline
1 hour&{\bf 24}&{\bf 40}&{\bf 55}&{\bf 58}\\
\hline
1 month&{\bf 31}&{\bf 47}&{\bf 61}&{\bf 65}\\
\hline
1 year&{\bf 33}&{\bf 49}&{\bf 64}&{\bf 67}\\
\hline
10 years&{\bf 36}&{\bf 51}&{\bf 66}&{\bf 70}\\
\hline
100 years&{\bf 38}&{\bf 54}&{\bf 68}&{\bf 72}\\
\hline
\end{array}
\newline
$
For the actual $Exp_{n0}$ finite complexity class, where the $ExpRank$ is allowed to reach $8*n0$, divide the values in the above table by 8. 

$\newline$
$
\begin{array}{|l|l|l|l|l|}
\hline
\bf{SemiPoly}&{SCC}&{SCS}&{MCT}&{MCA}\\
\hline
1 second&{\bf 35}&{\bf 179}&{\bf 568}&{\bf 725}\\
\hline
1 minute&{\bf 36}&{\bf 253}&{\bf 761}&{\bf 963}\\
\hline
1 hour&{\bf 86}&{\bf 353}&{\bf 1010}&{\bf 1264}\\
\hline
1 month&{\bf 162}&{\bf 580}&{\bf 1553}&{\bf 1923}\\
\hline
1 year&{\bf 201}&{\bf 693}&{\bf 1818}&{\bf 2241}\\
\hline
10 years&{\bf 245}&{\bf 814}&{\bf 2096}&{\bf 2578}\\
\hline
100 years&{\bf 295}&{\bf 955}&{\bf 2410}&{\bf 2953}\\
\hline
\end{array}
\newline
$

$\newline$
$
\begin{array}{|l|l|l|l|l|}
\hline
\bf{Poly}&{SCC}&{SCS}&{MCT}&{MCA}\\
\hline
1 second&{\bf 342}&{\bf 18*10^{3}}&{\bf 0.5*10^{6}}&{\bf 1*10^{6}}\\
\hline
1 minute&{\bf 1*10^{3}}&{\bf 45*10^{3}}&{\bf 1.2*10^{6}}&{\bf 2.4*10^{6}}\\
\hline
1 hour&{\bf 2.8*10^{3}}&{\bf 115*10^{3}}&{\bf 2.8*10^{6}}&{\bf 5.6*10^{6}}\\
\hline
1 month&{\bf 14*10^{3}}&{\bf 492*10^{3}}&{\bf 11*10^{6}}&{\bf 22*10^{6}}\\
\hline
1 year&{\bf 24*10^{3}}&{\bf 847*10^{3}}&{\bf 19*10^{6}}&{\bf 37*10^{6}}\\
\hline
10 years&{\bf 41*10^{3}}&{\bf 1.4*10^{6}}&{\bf 30*10^{6}}&{\bf 60*10^{6}}\\
\hline
100 years&{\bf 69*10^{3}}&{\bf 2.3*10^{6}}&{\bf 48*10^{6}}&{\bf 95*10^{6}}\\
\hline
\end{array}
\newline
$

$\newline$
$
\begin{array}{|l|l|l|l|l|}
\hline
\bf{Quadric}&{SCC}&{SCS}&{MCT}&{MCA}\\
\hline
1 second&{\bf 3162}&{\bf 9*10^{6}}&{\bf 13*10^{9}}&{\bf 71*10^{9}}\\
\hline
1 minute&{\bf 24*10^{3}}&{\bf 70*10^{6}}&{\bf 100*10^{9}}&{\bf 547*10^{9}}\\
\hline
1 hour&{\bf 190*10^{3}}&{\bf 0.55*10^{9}}&{\bf 0.77*10^{12}}&{\bf 4*10^{12}}\\
\hline
1 month&{\bf 5.1*10^{6}}&{\bf 15*10^{9}}&{\bf 21*10^{12}}&{\bf 114*10^{12}}\\
\hline
1 year&{\bf 18*10^{6}}&{\bf 51*10^{9}}&{\bf 72*10^{12}}&{\bf 396*10^{12}}\\
\hline
10 years&{\bf 56*10^{6}}&{\bf 162*10^{9}}&{\bf 229*10^{12}}&{\bf 1.3*10^{15}}\\
\hline
100 years&{\bf 177*10^{6}}&{\bf 511*10^{9}}&{\bf 723*10^{12}}&{\bf 3.9*10^{15}}\\
\hline
\end{array}
\newline
$

$\newline$
$
\begin{array}{|l|l|l|l|l|}
\hline
\bf{Linear}&{SCC}&{SCS}&{MCT}&{MCA}\\
\hline
1 second&{\bf 10^{7}}&{\bf 10^{14}}&{\bf 10^{20}}&{\bf 10^{21}}\\
\hline
1 minute&{\bf 10^{8}}&{\bf 10^{15}}&{\bf 10^{22}}&{\bf 10^{23}}\\
\hline
1 hour&{\bf 10^{10}}&{\bf 10^{17}}&{\bf 10^{23}}&{\bf 10^{25}}\\
\hline
1 month&{\bf 10^{13}}&{\bf 10^{20}}&{\bf 10^{26}}&{\bf 10^{28}}\\
\hline
1 year&{\bf 10^{14}}&{\bf 10^{21}}&{\bf 10^{27}}&{\bf 10^{29}}\\
\hline
10 years&{\bf 10^{15}}&{\bf 10^{22}}&{\bf 10^{28}}&{\bf 10^{30}}\\
\hline
100 years&{\bf 10^{15}}&{\bf 10^{23}}&{\bf 10^{29}}&{\bf 10^{31}}\\
\hline
\end{array}
\newline
$
The upper bound for the $PolyLog$ finite complexity class is in fact the same as for Linear. As such the above table applies to $PolyLog$ as well.
$\newline$

For $LogRank < log(n)/log(log(n))$ and $Const$ the growth rate allows inputs of almost any practical size to be solved in a very short amount of time, usually within much less than a second. Of course, sometimes in practice the exact $LogRank$ matters - for example when searching for a suitable value within an exponential universe of alternatives.

\end{document}